\documentclass[conference]{styles/IEEEtran}

\usepackage{calrsfs}
\DeclareMathAlphabet{\pazocal}{OMS}{zplm}{m}{n}
\usepackage{epsfig,endnotes,xspace,url,diagbox}
\AtBeginDocument{}
\usepackage{amsmath,amssymb}
\usepackage{amsthm}
\usepackage{dsfont}
\usepackage{tikz}
\usepackage{enumerate}
\usepackage{varioref}
\usepackage{graphicx}
\usepackage{epstopdf}
\usepackage{color}
\usepackage{multirow}
\usepackage{subfigure}
\usepackage{comment}
\usepackage[utf8]{inputenc}
\usepackage{mathtools}
\usepackage{wrapfig}
\usepackage[noend]{algpseudocode}
\usepackage{algorithm}
\usepackage{slashbox}
\usepackage{balance}
\usepackage{caption}
\usepackage{cryptocode}
\usepackage{booktabs}
\usepackage[pagebackref=false,breaklinks=true,colorlinks,bookmarks=false]{hyperref}
\hypersetup{
  colorlinks,
  linkcolor={blue!75!green},
  citecolor={red!70!orange},
  urlcolor={orange!70!red}
}

\newcommand{\mypara}[1]{\vspace*{0.05in}\noindent\textbf{#1.}$\;$}

\newtheorem{theorem}{Theorem}
\newtheorem{lemma}[theorem]{Lemma}

\newtheorem{prop}[theorem]{Proposition}
\newtheorem{cor}[theorem]{Corollary}
\newtheorem{definition}{Definition}

\newenvironment{proofsketch}{\par{\noindent \bf Proof Sketch:}}{\qed\\\par}

\newcommand{\EV}[1]{\ensuremath{\mathbb{E}\left[\, #1 \,\right]}}
\newcommand{\ind}[1]{\mathds{1}_{\left\{#1\right\}}}

\newcommand{\Var}[1]{\ensuremath{\mathsf{Var}\left[#1\right]}\xspace}
\newcommand{\var}{\ensuremath{\mathsf{Var}}\xspace}
\newcommand{\mse}{\ensuremath{\mbox{MSE}}\xspace}

\renewcommand{\Pr}[1]{\ensuremath{\mathsf{Pr} \left[#1\right] }\xspace}
\newcommand{\Prd}[2]{\ensuremath{\mathsf{Pr}^{#1} \left[#2\right] }\xspace}
\newcommand{\tuple}[1]{\ensuremath{\langle #1\rangle}\xspace}
\newcommand{\Domain}{\ensuremath{\mathcal{D}}\xspace}

\renewcommand{\AA}{\mathbf{A}}

\newcommand{\olh}{\ensuremath{\mathsf{OLH}}\xspace}
\newcommand{\had}{\ensuremath{\mathsf{Had}}\xspace}
\newcommand{\rap}{\ensuremath{\mathsf{RAP}}\xspace}
\newcommand{\rapr}{\ensuremath{\mathsf{RAP_R}}\xspace}

\newcommand{\grr}{\ensuremath{\mathsf{GRR}}\xspace}
\newcommand{\fo}{\ensuremath{\mathsf{FO}}\xspace}

\newcommand{\Adv}{\ensuremath{\mathsf{Adv}}\xspace}

\newcommand{\result}{\ensuremath{R}\xspace}
\newcommand{\results}{\ensuremath{\mathbf{R}}\xspace}
\newcommand{\myenc}{\ensuremath{\mbox{Enc}}\xspace}
\newcommand{\mydec}{\ensuremath{\mbox{Dec}}\xspace}

\newcommand{\MN}{\ensuremath{\mathsf{\mathsf{SH}}}\xspace}
\newcommand{\slh}{\ensuremath{\mathsf{SOLH}}\xspace}
\newcommand{\slhfull}{Shuffler-Optimal Local Hash\xspace}

\newcommand{\aue}{\ensuremath{\mathsf{AUE}}\xspace}
\newcommand{\methodseq}{\ensuremath{\mathsf{SS}}\xspace}
\newcommand{\methodsmc}{\ensuremath{\mathsf{PEOS}}\xspace}
\newcommand{\mursofull}{Private Encrypted Oblivious Shuffle\xspace}
\newcommand{\OSH}{\ensuremath{\mbox{oblivious shuffle}}\xspace}
\newcommand{\OSHA}{\ensuremath{\mbox{EOS}}\xspace}
\newcommand{\OSHAfull}{Encrypted Oblivious Shuffle\xspace}

\algnewcommand\User{\item[\textbf{User $i$:}]}
\algnewcommand\AuxServer{\item[\textbf{Shuffler $j$:}]}
\algnewcommand\AuxServerR{\item[\textbf{Shuffler $r$:}]}
\algnewcommand\AuxServerP{\item[\textbf{Shuffler $j\in [r-1]$:}]}
\algnewcommand\Server{\item[\textbf{Server:}]}
\algrenewcommand\algorithmicindent{1.0em}\newcommand{\mytitle}{Improving Utility and Security of the Shuffler-based Differential Privacy}

\begin{document}
\title{\mytitle}

\author{
	\IEEEauthorblockN{
		Tianhao Wang$^1$, Bolin Ding$^2$, Min Xu$^3$, Zhicong Huang$^2$, Cheng Hong$^2$, Jingren Zhou$^2$, Ninghui Li$^1$, Somesh Jha$^4$
	}
	\IEEEauthorblockA{
		$^1$Purdue University, $^2$Alibaba, $^3$University of Chicago, $^4$University of Wisconsin-Madison
	}
}

%\IEEEoverridecommandlockouts
%\makeatletter\def\@IEEEpubidpullup{6.5\baselineskip}\makeatother
%\IEEEpubid{\parbox{\columnwidth}{
%		Network and Distributed Systems Security (NDSS) Symposium 2021\\
%		21-24 February 2021, San Diego, CA, USA\\
%		ISBN 1-891562-61-4\\
%		https://dx.doi.org/10.14722/ndss.2021.23432\\
%		www.ndss-symposium.org
%	}
%	\hspace{\columnsep}\makebox[\columnwidth]{}}

\maketitle

\begin{abstract}
	When collecting information, local differential privacy (LDP) alleviates privacy concerns of users because their private information is randomized before being sent it to the central aggregator. LDP imposes large amount of noise as each user executes the randomization independently. To address this issue, recent work introduced an intermediate server with the assumption that this intermediate server does not collude with the aggregator. Under this assumption, less noise can be added to achieve the same privacy guarantee as LDP, thus improving utility for the data collection task.

	This paper investigates this multiple-party setting of LDP. We analyze the system model and identify potential adversaries. We then make two improvements: a new algorithm that achieves a better privacy-utility tradeoff; and a novel protocol that provides better protection against various attacks. Finally, we perform experiments to compare different methods and demonstrate the benefits of using our proposed method.
\end{abstract}

\section{Introduction}
\label{sec:intro}

To protect data privacy in the context of data publishing, differential privacy (DP)~\cite{Dwo06} is proposed and widely accepted as the standard of formal privacy guarantee.  DP mechanisms allow a server to collect users' data, add noise to the aggregated result, and publish the result.
More recently, local differential privacy (LDP) has been proposed~\cite{Duchi:2013:LPS}.  LDP differs from DP in that random noise is added by each user before the data is sent to the central server.  Thus, users do not need to trust the server.
This desirable feature of LDP has led to wider deployment by industry~\cite{rappor,apple-dp,ding2017collecting,sigmod:WangDZHHLJ19}.
Meanwhile, DP is still deployed in settings where the centralized server can be trusted (e.g., the US Census Bureau deployed DP for the 2020 census~\cite{uscensus}).
However, removing the trusted central party comes at the cost of utility.  Since every user adds some independently generated noise, the effect of noise adds up when aggregating the result.  As a result, while noise of scale (standard deviation) $\Theta(1)$ suffices for DP, LDP has noise of scale $\Theta(\sqrt{n})$ on the aggregated result ($n$ is the number of users). This gap is fundamental for eliminating the trust in the centralized server, and cannot be removed by algorithmic improvements~\cite{chan2019foundations}.

Recently, researchers introduced settings where one can achieve a middle ground between DP and LDP, in terms of both privacy and utility.  This is achieved by introducing an additional party~\cite{cheu2018distributed,erlingsson2019amplification,balle2019privacy,chowdhury2019outis}.  The setting is called the {\it shuffler model}.  In this model, each user adds LDP noise to data, encrypt it, and then send it to the new party called the shuffler.  The shuffler permutes the users' reported data, and then sends them to the server.  Finally the server decrypts the reports and obtains the result.  In this process, the shuffler only knows which report comes from which user, but does not know the content.  On the other hand, the server cannot link a user to a report because the reports are shuffled.
The role of the shuffler is to break the linkage between the users and the reports.  Intuitively, this anonymity can provide some privacy benefit.  {Therefore}, users can add less noise while achieving the same level of privacy.

In this paper, we study this new model from two perspectives.  First, we examine from the algorithmic aspect, and make improvement to existing techniques.
More specifically, in~\cite{balle2019privacy}, it is shown the essence of the privacy benefit comes from a ``noise'' whose distribution is independent of the input value, also called privacy blanket.  While existing work leverages this, it only works well when each user's value is drawn from a small domain.  To obtain a similar privacy benefit when the domain is large, we propose to use the local hashing idea (also considered in the LDP setting~\cite{Bassily2015local,wang2017locally,bassily2017practical,aistats:AcharyaSZ18}).  That is, each user selects a random hash function, and uses LDP to report the hashed result, together with the selected hash function.  By analyzing the utility and optimizing the parameters with respect to the utility metric (mean squared error), we present an algorithm that achieves accuracy orders of magnitude better than existing method.  We call {it} \slhfull (\slh).

We then work from the security aspect of the model.  We review the system setting of this model and identify two types of attack that were overlooked: collusion attack and data-poisoning attack.  Specifically, as there are more parties involved, there might exist collusions.  While existing work assumes non-collusion, we explicitly consider the consequences of collusions among different parties and propose a protocol \mursofull (\methodsmc) that is safer under these colluding scenarios.
The other attack considers the setting where the additional party introduces calibrated noise to bias the result or break the privacy protection.  To overcome this, our protocol \methodsmc takes advantage of cryptographic tools to prevent the shufflers from adding arbitrary noise.

To summarize, we provide a systematic analysis of the shuffler-based DP model.  Our main contributions are:

\begin{itemize}
	\item
	      We improve the utility of the model and propose \slh.

	\item
	      We design a protocol \methodsmc that provides better trust guarantees.

	\item
	      We provide implementation details and measure utility and execution performance of \methodsmc on real datasets.  Results from our evaluation are encouraging.
\end{itemize}

\section{Background}
\label{sec:back}

{We assume} each user possesses a value $v$ from a finite, discrete domain $\Domain$, and the goal is to estimate frequency of $v\in\Domain$.

\subsection{Differential Privacy}
\label{subsec:dp}
Differential privacy is a rigorous notion about individual's privacy in the setting where there is a trusted data curator, who gathers data from individual users, processes the data in a way that satisfies DP, and then publishes the results.  Intuitively, the DP notion requires that any single element in a dataset has only a limited impact on the output.

\begin{definition}[Differential Privacy] \label{def:dp}
	An algorithm $\AA$ satisfies $(\epsilon, \delta)$-DP, where $\epsilon, \delta \geq 0$,
	if and only if for any neighboring datasets $D$ and $D'$, and any set \results of possible outputs of $\AA$,
	\begin{equation*}
		\Pr{\AA(D)\in \results} \leq e^{\epsilon}\, \Pr{\AA(D') \in \results} + \delta
	\end{equation*}
\end{definition}
Denote a dataset as $D=\tuple{v_1, v_2, \ldots, v_n}$, where each $v_i$ is from some domain $\Domain$.  Two datasets $D=\tuple{v_1, v_2, \ldots, v_n}$ and $D'=\tuple{v'_1, v'_2, \ldots, v'_n}$ are said to be neighbors, or $D\simeq D'$, iff there exists at most one $i\in [n]=\{1,\ldots, n\}$ such that $v_i\neq v'_i$, and $v_j=v'_j$ for any other $j\neq i$.  When $\delta=0$, we simplify the notation and call $(\epsilon, 0)$-DP as $\epsilon$-DP.

\subsection{Local Differential Privacy}
\label{subsec:ldp}

Compared to the centralized setting, the local version of DP offers a stronger level of protection, because each user only reports the noisy data rather than the true data. Each user's privacy is still protected even if the server is malicious.

\begin{definition}[Local Differential Privacy] \label{def:ldp}
	An algorithm $\AA(\cdot)$ satisfies $(\epsilon, \delta)$-local differential privacy ($(\epsilon, \delta)$-LDP), where $\epsilon, \delta \geq 0$,
	if and only if for any pair of input values $v, v' \in \Domain$, and any set \results of possible outputs of $\AA$, we have
	\begin{equation*}
		\Pr{\AA(v)\in \results} \leq e^{\epsilon}\, \Pr{\AA(v')\in \results} + \delta
	\end{equation*}
\end{definition}

Typically, $\delta = 0$ in LDP ({thus} $\epsilon$-LDP).
We review the perturbation-based LDP mechanisms that will be used in the paper.

\mypara{Generalized Randomized Response}
\label{subsubsec:rr}
The basic mechanism in LDP is called randomized response~\cite{Warner65}.  It was introduced for the binary case (i.e., $\Domain=\{0,1\}$), but can be easily generalized.  Here we describe the generalized version of random response (\grr).

In \grr, each user with private value $v\in \Domain$ sends $\grr(v)$ to the server, where $\grr(v)$ outputs the true value $v$ with probability $p$, and a randomly chosen $v'\in \Domain$ where $v'\ne v$ with probability $1-p$.  Denote the size of the domain as $d=|\Domain|$, we have
\begin{align}
	\forall_{y \in \Domain}\;\Pr{\grr(v) = y}  = \left\{
	\begin{array}{lr}
		p=\frac{e^\epsilon}{e^\epsilon + d - 1}, & \mbox{if} \; y = v    \\
		q= \frac{1}{e^\epsilon + d - 1},         & \mbox{if} \; y \neq v \\
	\end{array}\label{eq:grr}
	\right.
\end{align}
This satisfies $\epsilon$-LDP since $\frac{p}{q}=e^\epsilon$.
To estimate the frequency of $\tilde{f}_v$ for $v\in \Domain$, one counts how many times $v$ is reported, denoted by $\sum_{i\in[n]}\ind{y_i=v}$, and then computes
\begin{align}
	\tilde{f}_v =  \frac{1}{n}\sum_{i\in[n]}\frac{\ind{y_i=v}-q}{p-q} \label{eq:grr_aggregate}
\end{align}
where $\ind{y_i=v}$ is the indicator function that tells whether the report of the $i$-th user $y_i$ equals $v$, and $n$ is the total number of users.

\mypara{Local Hashing}
When $d$ is large, the $p$ value in Equation~\eqref{eq:grr} becomes small, making the result inaccurate.  To overcome this issue, the local hashing idea~\cite{Bassily2015local} lets each user map $v$ to one bit, and then use \grr to perturb it.  More formally, each user reports $\tuple{H,\grr(H(v))}$ to the server, where $H$ is the mapping (hashing) function randomly chosen from a universal hash family.
In this protocol, both the hashing step and the randomization step result in information loss.  Later, Wang et al.~\cite{wang2017locally} realized $H$ does not necessarily hashes $v$ to one bit.  In fact, the output domain size $d'$ of $H$ is a tradeoff.  The optimal $d'$ is $e^\epsilon+1$.  And the method is called Optimized Local Hash (\olh).

Similar to \grr, the result of \olh needs to be calibrated.  Let $\tuple{H_i,y_i}$ be the report from the $i$'th user.
For each value $v\in \Domain$, to compute its frequency, one first computes $\sum_{i\in[n]}\ind{ H_i(v) = y_i}$ $=|\{i\mid H_i(v) = y_i\}|$,
and then computes
\begin{align}
	\tilde{f}_v = \frac{1}{n}\sum_{i\in[n]}\frac{\ind{ H_i(v) = y_i} - 1/d'}{p - 1/d'}\label{eq:olh_aggregate}
\end{align}

\subsection{Cryptographic Primitives}
\label{subsec:crypto_primitive}
We briefly review the cryptographic primitives that will be used.
Note that throughout this paper, we assume that the cryptographic tools are secure.

\mypara{Additive Homomorphic Encryption}
In Additive Homomorphic Encryption (AHE)~\cite{paillier1999public},
one can apply an algebraic operation (denoted by $\oplus$, e.g., multiplication) to two ciphertexts $c_1$, $c_2$, and get the ciphertext of the addition of the corresponding plaintexts.
More formally, there are two functions, encrypt function $\myenc$ and decrypt function $\mydec$.
Given two ciphertexts $c_1=\myenc(v_1)$ and $c_2=\myenc(v_2)$, we have $c_1\oplus c_2=\myenc(v_1+v_2)$.

\mypara{Additive Secret Sharing}
In this technique, a user splits a secret value $v\in\{0,\ldots, d-1\}$ into $r>1$ shares $\tuple{s_i}_{i\in [r]}$, where $r-1$ of them are randomly selected, and the last one is computed so that $\sum_i s_{i} \mod d = v$.  The shares are then sent to $r$ parties, so that each party only sees a random value, and $v$ cannot be recovered unless all the $r$ parties collaborate.

\mypara{Oblivious Shuffle}
In order to prevent the shuffler from knowing the mapping between the input and the output, oblivious shuffle introduces multiple shufflers.
A natural method is to connect the shufflers sequentially; and each shuffler applies a random shuffle.
Another way of achieving oblivious shuffle is the resharing-based shuffle~\cite{bogdanov2008sharemind,laur2011round} which utilizes secret sharing.
Suppose there are $r$ shufflers.  The users send their values to shufflers using secret sharing.  Define $t=\lfloor r/2 \rfloor + 1$ as the number of ``hiders'', and $r-t$ as the number of ``seekers''.  The resharing-based oblivious shuffle~\cite{laur2011round} proceeds like a ``hide and seek'' game.  In particular, there are ${r \choose t}$ partitions of the $r$ auxiliary servers into hiders and seekers.  For each partition, the seekers each splits its vector of shares into $t$ parts and sends them to the $t$ hiders, respectively.  Then the hiders accumulate the shares and shuffle their vectors using an agreed permutation.  The shuffled vectors are then split into $r$ shares and distributed to all of the $r$ auxiliary servers.  Note that now only the $t$ hiders know the permutation order.  The process proceeds for ${r \choose t}$ rounds to ensure that none of the colluding $r - t$ auxiliary servers know about the final permutation order.

\section{Problem Definition and Existing Techniques}
\label{sec:privacy_amp}

\subsection{Problem Definition}

Throughout the paper, we focus on the problem of histogram estimation, which is typically used for solving other problems in the LDP setting.  We assume there are $n$ users; each user $i$ possesses a value $v_i\in\Domain$.  The frequency of value $v\in \Domain$ is represented by ${f}_v = \frac{1}{n}\sum_{i\in[n]} \ind{v_i=v}$.  The server's goal is to estimate the frequency for each $v$, denoted by $\tilde{f}_v$.  The accuracy is measured by the mean squared error of the estimation, i.e., $\frac{1}{|\Domain|}\sum_{v\in\Domain}(f_v-\tilde{f}_v)^2$.

We consider the shuffler model, which is the middle ground between DP and LDP.  In particular, an auxiliary server called the shuffler is introduced.  Users need to trust that the auxiliary server does not collude with the original server.

\subsection{Privacy Amplification via Shuffling}
\label{subsec:amp}

The shuffling idea was originally proposed in Prochlo~\cite{bittau2017prochlo}, where a shuffler is inserted between the users and the server to break the linkage between the report and the user identification.  The privacy benefit was investigated in~\cite{cheu2018distributed,erlingsson2019amplification,balle2019privacy}.
It is proven that when each user reports the private value using \grr with $\epsilon_l$-LDP, applying shuffling ensures centralized $(\epsilon_c,\delta)$-DP, where $\epsilon_c<\epsilon_l$.
Table~\ref{tbl:mn_compare} gives a summary of these results.
Among them, \cite{balle2019privacy} provides the strongest result in the sense that the $\epsilon_c$ is the smallest, and the proof technique can be applied to other LDP protocols.

\begin{table}[!ht]
	\centering
	\resizebox{\columnwidth}{!}{
		\begin{tabular}{@{}c|c|c@{}}\toprule
			Method                             & Condition                                                     & $\epsilon_c$                                                    \\\midrule
			\cite{erlingsson2019amplification} & $ \epsilon_l < 1/2$                                           & $\sqrt{144\ln (1/\delta)\cdot \frac{\epsilon_l^2}{n}}$          \\ \hline
			\cite{cheu2018distributed}         & $\sqrt{\frac{192}{n}\ln (4/\delta)} < \epsilon_c < 1$, binary & $\sqrt{32\ln (4/\delta)\cdot \frac{e^{\epsilon_l}+1}{n}}$       \\ \hline
			\cite{balle2019privacy}            & $\sqrt{\frac{14\ln (2/\delta)d}{n-1}} < \epsilon_c \le 1$     & $\sqrt{14\ln (2/\delta)\cdot \frac{e^{\epsilon_l}+d-1}{(n-1)}}$ \\ \bottomrule
		\end{tabular}
	}
	\caption{Privacy amplification result comparison.  Each row corresponds to a method.  The amplified $\epsilon_c$ only differs in constants.  The circumstances under which the method can be used are different.  }\label{tbl:mn_compare}
\end{table}

\mypara{Recent Results}
Parallel to our work, \cite{Balcer2019,ghazi2019private} propose mechanisms other than \grr to improve utility in this model.  They both rely on the privacy blanket idea.  The method in~\cite{Balcer2019} gives better utility as it does not depend on $|\Domain|$.  However, the communication cost for each user is linear in $|\Domain|$, which is undesirable when $|\Domain|$ is large.  Moreover, its accuracy is worse than the method proposed in our paper.  We will analytically and empirically compare with~\cite{Balcer2019}.

\section{Improving Utility of the Shuffler Model}
\label{sec:slh}
We first review the intuition behind the privacy amplification proof, which is called the ``privacy blanket''.  Then we borrow the local hashing idea in LDP, and design a local hashing method that optimizes
accuracy in the shuffler setting.

\subsection{Privacy Blanket}
\label{subsec:blanket}

The technique used in~\cite{balle2019privacy} is called \textit{blanket decomposition}.  The idea is to decompose the probability distribution of an LDP report into two distributions, one dependent on the true value and the other independently random; and this independent distribution forms a ``privacy blanket''.
In particular, the output distribution of $\grr$ given in Equation~\eqref{eq:grr} is decomposed into
\begin{align*}
	\forall_{y \in \Domain}\;\Pr{\grr(v) = y}  = (1-\gamma)\Prd{v}{y}+\gamma\;\Pr{\mathsf{Uni}(\Domain) = y}
\end{align*}
where $\Prd{v}{y}$ is the distribution that depends on $v$, and $\mathsf{Uni}(\Domain)$ is uniformly random with $\Pr{\mathsf{Uni}(\Domain) = y}=1/d$.  With probability $1-\gamma$, the output is dependent on the true input; and with probability $\gamma$, the output is random.
Given $n$ users, the $n - 1$ (except the victim's) such random variables can be seen as containing some uniform noise (i.e., the $\gamma\; \Pr{\mathsf{Uni}(\Domain) = y}$ part).  For each value $v\in \Domain$, the noise follows $\mathsf{Bin}(n - 1, \gamma/d)$.  Intuitively, this noise makes the output uncertain.  The following theorem, which is derived from Theorem 3.1 of~\cite{balle2019privacy}, formalizes this fact.

\begin{theorem}[Binomial Mechanism]\label{thm:rr_blanket}
	Binomial mechanism adds independent noise $\mathsf{Bin}(n, p)$ to each component of the histogram.  It satisfies $(\epsilon_c, \delta)$-DP where
	\begin{align*}
		\epsilon_c = \sqrt{\frac{14\ln (2/\delta)}{np}}
	\end{align*}
\end{theorem}

In Theorem~\ref{thm:rr_blanket}, the larger $\gamma$ is, the better the privacy.  Given \grr, we can maximize $\gamma$ by setting $\Prd{v}{y} = \ind{v=y}$, which gives us $\gamma = \frac{d}{e^{\epsilon_l}+d-1}$.  The binomial noise $\mathsf{Bin}(n - 1, \frac{1}{e^{\epsilon_l}+d-1})$ thus provides $(\sqrt{14\ln (2/\delta)\cdot \frac{e^{\epsilon_l}+d-1}{(n-1)}}, \delta)$-DP~\cite{balle2019privacy}.
One limitation of~\cite{balle2019privacy} is that
as \grr is used, the accuracy downgrades with domain size $d$.

\subsection{\slhfull}
\label{subsec:slh}

In order to benefit from the shuffler model in the case when the domain size $d$ is large, the key is to derive a mechanism whose utility does not degrade with $d$.

\subsubsection{Unary Encoding for Shuffling}
\label{sec:unary_encoding}
We first revisit the unary-encoding-based methods, also known as the basic RAPPOR~\cite{rappor}, and show that this class of methods can enjoy the benefit of the privacy blanket argument.
In particular, in unary-encoding, the value $v$ is transformed into a vector $B$ of size $d$, where $B[v]=1$ and the other locations of $B$ are zeros (note that this requires values of the domain $\Domain$ be indexed from $1$ to $d$).  Then each bit $b$ of $B$ is perturbed to $1-b$ independently.
To satisfy LDP, the perturbation probability is set to $\frac{1}{e^{\epsilon/2}+1}$.  Note that we use $\epsilon/2$ because for any two values $v$ and $v'$, their corresponding unary encodings differ by two bits.   We can apply the privacy blanket argument and prove that a $\epsilon_l$-LDP unary-encoding method satisfies $(\epsilon_c,\delta)$-DP after shuffling.

\begin{theorem}
	\label{thm:sue_privacy}
	Given an $\epsilon_l$-LDP unary-encoding method, after shuffling, the protocol is $(\epsilon_c,\delta)$-DP, where
	\begin{align}
		\epsilon_c & = 2\sqrt{14\ln (4/\delta) \cdot \frac{e^{\epsilon_l/2} + 1}{n-1}}\nonumber
	\end{align}
\end{theorem}

\begin{proof}
	For any two neighboring datasets $D\simeq D'$, w.l.o.g., we assume they differ in the $n$-th value, and $v_n=1$ in $D$, $v_n=2$ in $D'$.  By the independence of the bits, probabilities on other locations are equivalent.  Thus we only need to examine the summation of bits for location $1$ and $2$.  For each location, there are $n-1$ users, each
	reporting the bit with probability
	\begin{align*}
		\forall_{y \in \{0, 1\}}\;\Pr{B[j] \rightarrow y}  = (1-\gamma)\ind{B[j]= y}+\gamma\;\Pr{\mathsf{Uni}(2) = y}
	\end{align*}
	where we slight abuse the notation and use $\mathsf{Uni}(2)$ for $\mathsf{Uni}(\{0,1\})$.  Given that the perturbation probability is $\Pr{1 \rightarrow 0} = \Pr{0 \rightarrow 1} = \frac{1}{e^{\epsilon_l/2}+1} = \gamma/2$, we can calculate that $\gamma = \frac{2}{e^{\epsilon_l/2} + 1}$.
	After shuffling, the histogram of $n - 1$ (except the victim's) such random variables follows $\mathsf{Bin}(n - 1, \gamma/2)$.
	As there are two locations, by Theorem~\ref{thm:rr_blanket}, we have
	$\epsilon_c = 2\sqrt{14\ln (4/\delta) \cdot \frac{e^{\epsilon_l/2} + 1}{n-1}}$.
\end{proof}

\subsubsection{Local Hashing for Shuffling}
While sending $B$ when $d$ is large is fine for each user; with $n$ users, receiving $B$'s from the server side is less tolerable as it incurs $O(d\cdot n)$ bandwidth.  To reduce the communication cost, we propose a hashing-based method, with a tradeoff between computation and communication.  From the server side, it requires more computation cost than the unary-encoding based methods; but the overall communication bandwidth is smaller.  In what follows, we prove the hashing-based method is private in the shuffler model.

We remind the readers that in local hashing, each user reports $H$ and $y= \grr(H(v))$.  The hash function $H$ is chosen randomly from a universal hash family and hashes $v$ from a domain of size $d$ into another domain of size $d'\leq d$; and $\grr$ will report $H(v)$ with probability $\frac{e^{\epsilon_l}}{e^{\epsilon_l}+d'-1}$, and any other value (from the domain of size $d'$) with probability $\frac{1}{e^{\epsilon_l}+d'-1}$ (Equation~\eqref{eq:grr}).
In terms of blanket decomposition, the user reports truthfully with probability $1-\gamma = \frac{e^{\epsilon_l}-1}{e^{\epsilon_l}+d'-1}$; and if the user reports randomly, any value from $[d']$ can be reported with equal probability.
We call this method \slh, which stands for \slhfull.

\begin{theorem}
	\label{thm:olh_privacy}
	Given the $\epsilon_l$-LDP \slh method, after shuffling, the protocol is $(\epsilon_c, \delta)$-DP, where
	\begin{align}
		\epsilon_c & =\sqrt{\frac{14\ln (2/\delta)(e^{\epsilon_l} + d' - 1)}{n - 1}} \nonumber
	\end{align}
\end{theorem}
The proof of this theorem is technically challenging and is highly non-trivial.  But due to the space limit, we leave the full proof to the appendix and  provide the high-level idea as follows:
We first assume that the server knows which users other than the victim (user $n$) report truthfully (i.e., with probability $1-\gamma = \frac{e^{\epsilon_l}-1}{e^{\epsilon_l}+d'-1}$), and prove that the server can delete these reports from the shuffled reports.  For the remaining reports, we then prove the probability ratio can be simplified to ratio of two Binomial random variables.  Finally, we bound this ratio and obtain $\epsilon_c$ and $\delta$.

\subsubsection{Utility Analysis}
\label{subsec:utility_basic}
Now we analyze the utility of different methods.  We utilize the framework of Theorem 2 from~\cite{wang2017locally} to analyze the accuracy of estimating the frequency of each value in the domain (i.e., Equations~\eqref{eq:grr_aggregate} and~\eqref{eq:olh_aggregate}).
In particular, we measure the expected squared error of the estimation $\tilde{f}_v$, which equals variance, i.e.,
\begin{align*}
	\sum_{v\in \Domain}\EV{(\tilde{f}_v - f_v)^2} = \sum_{v\in \Domain}\Var{\tilde{f}_v}
\end{align*}
Fixing the local $\epsilon_l$, the variances are already summarized in~\cite{wang2017locally}; our analysis extends that into the shuffler setting.  We fix $\epsilon_c$ and estimate variance for different methods.

\mypara{Utility of Generalized Randomize Response}
We first prove the variance of \grr.
\begin{prop}
	\label{thm:utility_grr}
	Given $\epsilon_c$ in the shuffler model, the variance of using \grr is bounded by
	$\frac{\frac{\epsilon_c^2(n - 1)}{14\ln (2/\delta)} - 1}{n\left(\frac{\epsilon_c^2(n - 1)}{14\ln (2/\delta)} - d\right)^2} $.
\end{prop}
\begin{proof}
	Given the domain size $d$ and the LDP parameter $\epsilon_l$, the variance is given in~\cite{wang2017locally}.
	Here for completeness, we present the full proof.  We will omit these steps in the following proofs.
	Denote $p=\frac{e^{\epsilon_l}}{e^{\epsilon_l} + d- 1}$, $q = \frac{1}{e^{\epsilon_l} + d- 1}$, and $y_i$ is the report of user $i$, we have
	\begin{align*}
		\Var{\tilde{f}_v} = & \Var{\frac{1}{n}\left(\sum_{i\in[n]}\frac{\ind{v = y_i}-q}{p-q} \right)} \\
		=                   & \frac{1}{n^2}\Var{\sum_{i\in[n]}\frac{\ind{v = y_i}}{p - q}}             \\
		=                   & \frac{\sum_{i\in[n]}\Var{\ind{v = y_i}}}{n^2\cdot (p - q)^2}
	\end{align*}
	Here for each of the $n$ users, if the true value is $v$ (there are $nf_v$ of them) we have $\Var{\ind{v = y_i}} = p(1-p)$; otherwise, we have $\Var{\ind{v = y_i}} = q(1 - q)$ for the rest $n(1-f_v)$ users.  Together, we have
	\begin{align*}
		\Var{\tilde{f}_v} = & \frac{nf_v p(1-p) + n(1-f_v) q\left(1-q\right)}{n^2 (p-q)^2}          \\
		=                   & \frac{ q(1-q)}{n(p-q)^2} + \frac{f_v \left(1 - p - q\right)}{n (p-q)}
	\end{align*}
	Plugging in the value of $p$ and $q$, and assuming $f_v$ is small on average, then we have
	\begin{align*}
		\Var{\tilde{f}_v} \leq \frac{ q(1-q)}{n(p-q)^2} = \frac{e^{\epsilon_l} + d- 2}{n(e^{\epsilon_l}- 1)^2}
	\end{align*}
	From~\cite{balle2019privacy}, we have
	$e^{\epsilon_l} + d - 1 = \frac{\epsilon_c^2(n - 1)}{14\ln (2/\delta)}$.  Thus the variance becomes
	$\frac{\frac{\epsilon_c^2(n - 1)}{14\ln (2/\delta)} - 1}{n\left(\frac{\epsilon_c^2(n - 1)}{14\ln (2/\delta)} - d\right)^2}$.
\end{proof}

\mypara{Utility of Unary Encoding (RAPPOR)}
Similarly, we can prove the variance of unary encoding.
\begin{prop}
	\label{thm:utility_ue}
	Given $\epsilon_c$ in the shuffler model, the variance of using unary encoding (RAPPOR) is bounded by
	$\frac{\frac{\epsilon_c^2(n - 1)}{56\ln (4/\delta)} - 1}{n\left(\frac{\epsilon_c^2(n - 1)}{56\ln (4/\delta)} - 2\right)^2}$.
\end{prop}
\begin{proof}
	According to~\cite{wang2017locally}, the variance of RAPPOR given $\epsilon_l$ is
	\begin{align*}
		\frac{e^{\epsilon_l / 2}}{n(e^{\epsilon_l / 2}  - 1)^2}
	\end{align*}
	From Theorem~\ref{thm:sue_privacy}, we have
	$e^{\epsilon_l / 2} + 1 = \frac{\epsilon_c^2(n - 1)}{56\ln (4/\delta)}$.  Thus the variance becomes $\frac{\frac{\epsilon_c^2(n - 1)}{56\ln (4/\delta)} - 1}{n\left(\frac{\epsilon_c^2(n - 1)}{56\ln (4/\delta)} - 2\right)^2}$.
\end{proof}

\mypara{Utility of Local Hashing}
Now we prove the variance of \slh and instantiate $d'$.
\begin{prop}
	\label{thm:utility_slh}
	Given $\epsilon_c$ in the shuffler model, the variance of using \slh is bounded by
	$\frac{\left(\frac{\epsilon_c^2(n - 1)}{14\ln (2/\delta)}\right)^2}{n\left(\frac{\epsilon_c^2(n - 1)}{14\ln (2/\delta)} - d'\right)^2(d' - 1)}$.
\end{prop}
\begin{proof}
	According to Equation (10) of~\cite{wang2017locally}, the variance of local hashing given $\epsilon_l$ is
	\begin{align}
		\frac{(e^{\epsilon_l} + d' - 1)^2}{n(e^{\epsilon_l}  - 1)^2(d' - 1)}\label{eq:var_lh}
	\end{align}
	From Theorem~\ref{thm:olh_privacy}, we have
	$e^{\epsilon_l} + d' - 1 = \frac{\epsilon_c^2(n - 1)}{14\ln (2/\delta)}$.  Thus the variance becomes
	$\frac{\left(\frac{\epsilon_c^2(n - 1)}{14\ln (2/\delta)}\right)^2}{n\left(\frac{\epsilon_c^2(n - 1)}{14\ln (2/\delta)} - d'\right)^2(d' - 1)} $.
\end{proof}

\mypara{Optimizing Local Hashing}
Note that $d'$ is unspecified.  We can tune $d'$ to optimize variance given a fixed $\epsilon_c$.  Denote $m$ as $\frac{\epsilon_c^2(n - 1)}{14\ln (2/\delta)}$, our goal is to choose $d'$ that minimize this variance $\var(m, d') = \frac{m^2}{n(m-d')^2(d'-1)}$.
By making its partial derivative to $0$, we can obtain that when
\begin{align}
	d'=\frac{m+2}{3}=\frac{\epsilon_c^2(n - 1)}{42\ln (2/\delta)} + \frac{2}{3}\label{eq:opt_d}
\end{align}
the variance is minimized.
Note that $d'$ can only be an integer.
In the actual implementation, we choose $d'$ to be $\lfloor(m+2)/3\rfloor$.
Thus the variance is optimized to $\var(m, \lfloor(m+2)/3\rfloor)$.

\mypara{Comparison of the Methods}
We first observe that the variance of \grr grows with $d$ (as shown in Proposition~\ref{thm:utility_grr}).  When $d$ is large, we should use unary encoding or local hashing.  Between the two, the variance of unary encoding is slightly better, however, its communication cost is higher.
Thus, between \grr and \slh, we can choose the one with better utility by comparing Proposition~\ref{thm:utility_grr} and $\var(m, \lfloor(m+2)/3\rfloor)$.

\subsubsection{Comparison with Parallel Work}
Parallel to our work, \cite{Balcer2019,ghazi2019private} also propose mechanisms to improve utility in this model.
Among them \cite{Balcer2019} gives better utility which does not depend on $|\Domain|$.  Similar to our method, its proof also utilizes Theorem~\ref{thm:rr_blanket}.  But the approach is different.  In particular, \cite{Balcer2019} first transforms the data using one-hot encoding, then independently increment values in each location with probability $p = 1 - \frac{200}{\epsilon_c^2n}\ln(4/\delta)$.  We call this method \aue for appended unary encoding.  As each location is essentially a Bernoulli bit, its variance is
$p(1-p) = \frac{200}{\epsilon_c^2n}\ln(4/\delta)\left(1 - \frac{200}{\epsilon_c^2n}\ln(4/\delta)\right)$.
Compared with Lemma~\ref{thm:utility_slh}, this gives comparable results (differing by only a constant).  But this protocol itself is not LDP.  Moreover, as one-hot encoding is used, the communication cost for each user is linear in $|\Domain|$, which is even worse than \grr.  We will empirically compare with~\cite{Balcer2019} in the experimental evaluation section.

More recently, \cite{arxiv:erlingsson2020encode} also proposed a similar unary-encoding-based method.
We note that~\cite{arxiv:erlingsson2020encode} operate on a novel removal LDP notion.  More specifically, previous (ours included) LDP and shuffler-based LDP literature works with Definition~\ref{def:ldp}, which ensures that for each user, if his/her value changes, the report distribution is similar.  \cite{arxiv:erlingsson2020encode} introduces a novel removal LDP notion inspired by the removal DP.  In particular, removal DP states that for any two datasets $D$ and $D_-$, where $D_-$ is obtained by removing any one record from $D$, the output distributions are similar.  Extending that idea to the local setting, removal LDP states that for each user, whether his/her value is empty or not, the report distribution is similar.  Given that, a unary-encoding-based method similar to RAPPOR~\cite{rappor} is proposed.  The method is similar to the method we described in Section~\ref{sec:unary_encoding}, except that privacy budget $\epsilon_l$ is not divided by $2$.  Interestingly, any $\epsilon$-{Removal LDP} algorithm is also a $2\epsilon$-{Replacement LDP} algorithm, because
\[
	\Pr{\AA(v)\in \results} \leq e^{\epsilon} \Pr{\AA(\bot)\in \results}\leq e^{2\epsilon} \Pr{\AA(v')\in \results}
\]
where $\bot$ is a special ``empty'' input.  As a result, in our LDP setting, the two methods achieves the same utility.

\section{Security Analysis}
\label{sec:system}
This section focuses on the analyzing the security implications of the shuffler model.  We identify different parties and potential attacks.  Then we propose countermeasures using secret sharing and oblivious shuffle in next section.

\subsection{Parties and Attackers}
There are three types of parties in the shuffler model: \textit{users}, the \textit{server}, and the \textit{auxiliary servers} (shufflers).
The auxiliary servers do not exist in the traditional models of DP and LDP; and in DP, the server may share result with some external parties.
Figure~\ref{fig:model} provides an overview of the system model.

\mypara{The Attackers}
From the point of view of a single user, other parties, including the auxiliary server, the server, and other users, could all be adversaries.
We assume all parties have the same level of background knowledge, i.e., all other users' information except the victim's.  This assumption essentially enables us to argue DP-like guarantee for each party.

The prominent adversary is the server.  Other parties can also be adversaries but are not the focus because they have less information.
For example, in the shuffler-based approach, there is only one auxiliary server.  It knows nothing from the ciphertext.

\mypara{Additional Threat of Collusion}
We note that in the multi-party setting, one needs to consider the consequences when different parties collude.  In general, there are many combinations of colluding parties.  And understanding these scenarios enables us to better analyze and compare different approaches.

In particular, the server can collude with the auxiliary servers.  If all the auxiliary servers are compromised, the model is reduced to that for LDP.  Additionally, the server can also collude with other users (except the victim), but in this case the model is still LDP.
On the other hand, if the server only colludes with other users, it is less clear how the privacy guarantee will downgrade.
Other combinations are possible but less severe.  Specifically, there is no benefit if the auxiliary servers collude with the users.
We consider all potential collusions and highlight three important (sets of) adversaries:

\begin{itemize}
	\item
	      $\Adv$: the server itself.

	\item
	      $\Adv_u$: the server colluding with other users.

	\item
	      $\Adv_a$: the server with the auxiliary servers.
\end{itemize}

\begin{figure}
	\centering
	\subfigure{
		\includegraphics[width=0.4\textwidth]{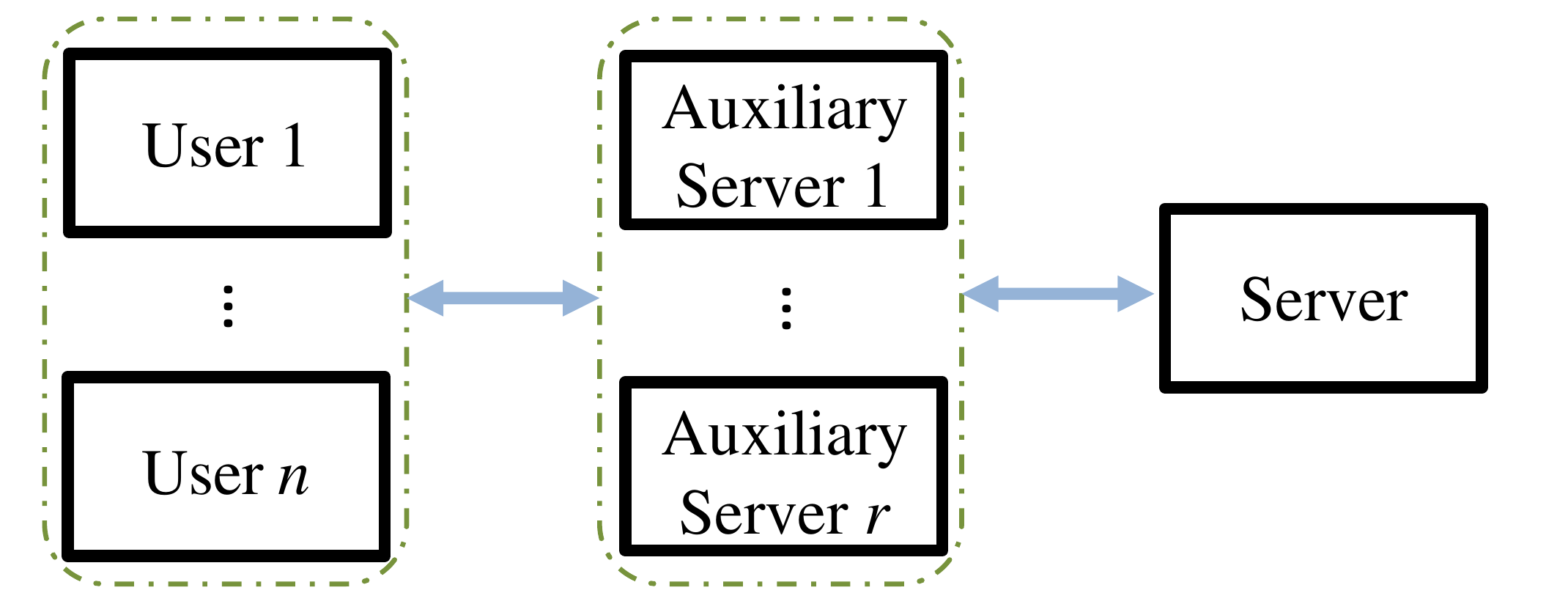}
	}
	\caption{Overview of parties and interactions.  Users communicate with the auxiliary servers.  The auxiliary servers processes the users' data, and communicate with the server.
	}
	\label{fig:model}
\end{figure}

\subsection{Privacy Guarantees of Existing Methods}
\label{subsec:existing_privacy}
Having identified the potential adversaries and the proving technique, now we examine the shuffler-based DP.  The key ideas are (1) We model each attack's view using an algorithm, such that we can prove the DP guarantee. (2) We prove the DP guarantee for each party separately. Existing work focuses on $\Adv$, but we examine the privacy guarantee against each of the $\Adv$'s.  This gives a comprehensive understanding of the system's privacy guarantee.

In particular, existing work showed that if each user executes an $\epsilon_l$-LDP protocol, the view of $\Adv$ is $(\epsilon_c, \delta)$-DP.
If the users collude with the server, the server's view is composed of two parts: the shuffled reports as in $\Adv$, and all users' reports except the victim's.  By subtracting each user's reports from the shuffled result, the server now knows the victim's LDP report; thus the model falls back to the original setting.  Finally, if the shuffler colludes with the server, the model also degrade to the LDP setting.

Note that we assume the cryptographic primitives are safe (i.e., the adversaries are computationally bounded and cannot learn any information from the ciphertext) and there are no side channels such as timing information.
In some cases, the whole procedure can be interactive, i.e., some part of the observation may depend on what the party sends out.  For this, one can utilize composition theorems to prove the DP guarantee.
Moreover, the parties are assumed to follow the protocol in the privacy proofs.
If the parties deviate from the prescribed procedure, we examine the possible deviations and their influences in the next subsection.

\subsection{Robustness to Malicious Parties}
\label{subsec:robust}

There could be multiple reasons for each party to be malicious to (1) interrupt the data collection process, (2) infer more sensitive information from the users, and (3) degrade the utility (estimation accuracy) of the server.  In what follows, for each of the reasons, we analyze the consequence and potential mitigation of different parties.  Note that the server will not deviate from the protocol as it is the initiator, unless to infer more information of the users.

First, any party can try to interrupt the process; but it is easy to mitigate.  If a user blocks the protocol, his report can be ignored.  If the auxiliary server denies the service, the server can find another auxiliary server and redo the protocol.  Note that in this case, users need to remember their report to avoid averaging attacks.

Second, it is possible that the auxiliary server deviates from the protocol (e.g., by not shuffling LDP reports), thus the server has access to the raw LDP reports.  In these cases, the server can learn more information, but the auxiliary server does not have benefits except saving some computational power.  And if the auxiliary server colludes with the server, they can learn more information without any deviation.  Thus we assume the auxiliary server will not deviate in order to infer sensitive information.
For the server, as it only sees and evaluates the final reports; and the reports are protected by LDP, there is nothing the server can do to obtain more information from the users.

Third, we note that any party can degrade the utility.  Any party other than the server has the incentive to do this.
For example, when the server is interested in learning the popularity of websites, different parties can deviate to promote some targeted website.  This is also called the data poisoning attack.
To do this, the most straight-forward way is to generate many fake users, and let them join the data collection process.  This kind of Sybil attack is hard to defend against without some kind of authentication, which is orthogonal to the focus of this paper.
Each user can change the original value or register fake accounts; and this cannot be avoided.  But any ability beyond it is undesirable.
In addition, the protocol should restrict the impact of the auxiliary server on the result.

To summarize, different parties can deviate from the protocol, but we argue that in most cases, a reasonable party has no incentive to do this, other than poisoning the result.  We are mainly concerned about the users or the auxiliary server disrupting utility.

\subsection{Discussion and Key Observations}
\label{subsec:key_obs}

In this section, we first systematically analyze the setting of the shuffler-based DP model.  In addition to the adversary of the server, we highlight two more sets of adversaries.  We then propose to analyze the privacy guarantee against different (sets of) adversaries.  Finally, we discuss the potential concern of malicious parties.  Several observations and lessons are worth noting.

\mypara{When Auxiliary Server Colludes: No Amplification}
When the server colludes with the auxiliary servers, the privacy guarantee falls back to the original LDP model.  When using the shuffler model, we need to reduce the possibility of this collusion, e.g., by introducing more auxiliary servers.

\mypara{When Users Collude: Possibility Missed by Previous Literature}
When proving privacy guarantees against the server, existing work assumes the adversary has access to users' sensitive values but not the LDP output.  While this is possible, we note that if an adversary already obtains users' sensitive values, it may also have access to the users' LDP reports.  Such cases include the users (except the victim) collude with the server; or the server is controlling the users (except the victim).
Thus, the assumption in the shuffle-based amplification work uncommon in real-world scenarios, which makes the privacy guarantee less intuitive to argue.

\mypara{When Parties Deviates: Avoid Utility Disruption}
The protocol should be designed so that each individual user or auxiliary server has limited impact on the estimation result.

\begin{figure*}
	\centering
	\subfigure{
		\includegraphics[width=0.98\textwidth]{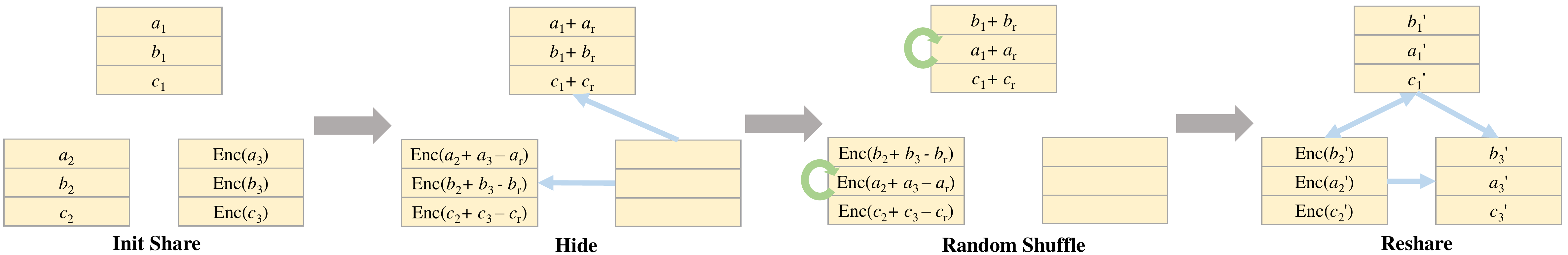}
	}
	\caption{Overview of \OSHA with $r=3$ shufflers and $n=3$ values $a, b, c$.  Each shuffler receives $n$ shares; and one shuffler's shares are encrypted by additive homomorphic encryption.  During hiding, one shuffler sends its shares to the other two shufflers, who then shuffle the aggregated shares with an agreed permutation.  To reshare, each of the shufflers splits its shares and send them to the other shufflers.
	}
	\label{fig:osha}
\end{figure*}

\section{Defending against Attacks}
\label{sec:proposal}
We present a protocol that improves the security guarantee of existing work.  The goal is to simultaneously defend against three threats: (1) the server colludes with the users; (2) the server colludes with the auxiliary servers; (3) data poisoning from each party.

\subsection{Fake Response from Auxiliary Servers}
\label{sec:fake_response}

To defend against the threat when the server colludes with the users, we propose to have the auxiliary servers inject noise.
There can be different ways to do this.
Our approach utilizes uniform fake reports.  The intuition of this approach is that (1) its analysis is compatible with the privacy blanket argument, which will be more clear later; and (2) the expected noise for each value in the domain is the same, thus suitable for obtaining a good privacy amplification effect.
On the server side, after obtaining the estimated frequency $\tilde{f}$, the server recovers the frequency for the original dataset by subtracting the expected noise, i.e.,
\begin{align}
	f'_v = \frac{n + n_r}{n}\tilde{f}_v - \frac{n_r}{n}\frac{1}{d} \label{eq:murs_fo}
\end{align}

Building on top of this, we present efforts to defend against the other two threats, i.e., the server colluding with the auxiliary servers, and data poisoning attack.

\subsubsection{First Attempt: Sequential Shuffle}
\label{subsec:ss}
To improve the trust model of the shuffler-based model, one idea is to introduce a sequence of shufflers, so that as long as one shuffler is trusted, the privacy guarantee remains.  In this case, the task of inserting $n_r$ fake reports can be divided equally among the $r$ auxiliary servers (shufflers).
More specifically, the first shuffler receives the users' LDP reports as input, and draws $n_u = n_r / r$ fake reports.  It then shuffles all the reports and sends them to the second shuffler, who draws another $n_u$ fake reports, shuffles all the reports, and sends them to the next shuffler.  This procedure proceeds until the last shuffler sends the result to the server.  Onion encryption is used during the process; each party decrypts one layer of encryption, and the server obtains $n + n_r$ reports.

However, this approach is vulnerable to poison attacks by the shufflers.  That is, the auxiliary servers can replace the users' reports with any report of their choice to change the final result, and the fake reports each shuffler inserts can be chosen arbitrarily.

To mitigate the first threat, we can use an idea of spot-checking.  That is, the server can add dummy accounts before the system setup, then it can check whether the reports from his accounts are tampered.  For the second threat, we find that it hard to handle.  Specifically, a dishonest auxiliary server may draw fake reports from some skewed (instead of uniform) distribution in order to mislead the analyzer and achieve a desired result; and there is no way to self-prove the randomness he used is truly random.

\subsubsection{Second Attempt: Oblivious Shuffle}
\label{subsec:os}

To overcome the data poisoning attack, our approach is to construct the fake reports using secret sharing, which ensures that as long as one shuffler is honest, the inserted fake reports are uniformly random.
To share an LDP report, we note that for both \grr and \slh, the domain of the report can be mapped to an ordinal group $\{0, 1, \ldots, x\}$, where each index represents one different LDP report.  Thus the LDP reports can be treated as numbers and shared with additive secret sharing.

In order to shuffle shares of secret, we utilize the oblivious shuffle protocol described in Section~\ref{subsec:crypto_primitive}.
More specifically, the $n$ users each splits his/her LDP reports into $r$ shares among the $r$ shufflers.  Each of the shufflers then uniformly draws one share for each of the $n_r$ fake reports.  Thus the shufflers each has $n+n_r$ shares; and the sums of the shares equal to the $n$ reports from users and $n_r$ report that are random.  An oblivious shuffle protocol is then executed among the shufflers to shuffle the $n + n_r$ shares of reports.
Finally the $r$ shufflers send their shares to the server, who combines the shares to obtain the results.
Note that the communication is assumed to be done via secure channels.

This solution suffers from a threat that, even without the server, half of the shufflers can collude to recover the user reports.
To mitigate this concern, we design a new oblivious shuffle protocol \OSHA that uses additive homomorphic encryption (AHE).

\subsubsection{Proposal: \mursofull}
\label{subsec:eos}
To ensure that the shufflers cannot infer the users' reported data, a natural solution is to encrypt the shares using the server's public key.  Moreover, the encryption needs to be additively homomorphic in order to be compatible with the secret-sharing operations.
In what follows, we present a new protocol \OSHAfull (\OSHA) that utilizes additive homomorphic encryption (AHE) in oblivious shuffle.  We then present our proposal \mursofull (shorted for \methodsmc) that uses \OSHA for DP.

\mypara{\OSHAfull}
\OSHAfull (\OSHA) works similarly to oblivious shuffle.  One difference is that in each round, one shuffler will possess the encrypted shares.  The encrypted shares can be shuffled and randomized just like normal shares except that they are then processed under AHE.

Denote the shuffler who possess encrypted shares as $E$.  In each round, $E$ splits its encrypted vector of shares into $t$ new vectors so that $t-1$ of which are plaintexts, and the last one is still in the ciphertext form (this can be done because of AHE).  The $t$ shares are randomly sent to the $t$ hiders.  Only one of them will receive the ciphertext share and become the next $E$.  After the group shuffling, the new $E$ splits its vector of shares and sends them to $r$ parties.  An example of \OSHA with $r=3$ is demonstrated in Figure~\ref{fig:osha}.
\OSHA strengthens \OSH in that even if the $r$ shufflers collude, they cannot figure out the users' original reports, because one share is encrypted.

Note that there is a crucial requirement for the AHE scheme: it should support a plaintext space of $\mathbb{Z}_{2^\ell}$ where $\ell$ is normally 32 or 64 in our case.  This is because the fake reports are sampled locally as random $\ell$-bit shares, and later they will be encrypted and added in AHE form, so that the decrypted result modulo $2^\ell$ looks like other reports.  Otherwise the fakeness will be detected by the server. Such an AHE scheme can be instantiated to be the full-decryption variant of DGK~\cite{damgard2008hom} using Pohlig-Hellman algorithm~\cite{pohlig1978improved}.

\begin{cor}
	\label{cor:osha}
	Encrypted oblivious shuffle, instantiated with additive homomorphic encryption of plaintext space $\mathbb{Z}_{2^\ell}$, is a secure oblivious shuffle protocol in the semi-honest model.
\end{cor}
\begin{proofsketch}
	The difference of \OSHA from \OSH is that AHE is used for one hider's computation in each round.  As long as AHE does not leak additional information, similar proof about the final shuffling order can be derived from \OSH~\cite{laur2011round}.

	For AHE, note that although we use AHE for one hider's computation in each round, the computation is translated into modulo $2^\ell$ in the plaintext space, which is exactly the same as normal secret sharing computation. Therefore, AHE does not leak additional information as long as the security assumption of the AHE holds (hardness of integer factorization in the case of DGK).
\end{proofsketch}

\mypara{Using \OSHA for Differential Privacy}
To use \OSHA for DP, each user encrypts one share (w.l.o.g., the $r^{th}$ share) using the server's public key $pk_s$ before uploading.  In addition, we have the shufflers add fake reports.  The full description of this protocol is given in Algorithm~\ref{MPC}.
There are three kinds of parties, users, shufflers, and the server.  They all agree to use some method \fo with the same parameter (e.g., $\epsilon$, domain size, etc); the \fo can be either \grr or \slh, depending on the utility, as described in Section~\ref{subsec:utility_basic}.  All the communication is done through a secure channel.  The users split their LDP reports into $r$ shares, encrypt only the $r$-th shares using AHE, and send them to the shufflers.  Each shuffler generate $n_r$ shares for fake reports; only the $r$-th shuffler encrypt the shares with AHE.  In this case, a malicious shuffler can draw its shares from a biased distribution; but those shares will then be ``masked'' by other honest shufflers' random shares and become uniformly random. By Corollary~\ref{cor:osha}, the users' reports are protected from the shufflers; and the server cannot learn the permutation unless he can corrupt more than half of the auxiliary servers.

\begin{algorithm}[t]
	\caption{\methodsmc}
	\label{MPC}
	\begin{algorithmic}[1]
		\User{Value $v_i$}
		\State $Y_i = \fo(v_i)$ \Comment{\fo can be \grr or \slh}
		\State Split $Y_i$ into $r$ shares $\tuple{Y_{i, j}}_{j\in [r]}$
		\For{$j\in[r-1]$}
		\State Send $Y_{i, j}$ to auxiliary server $j$
		\EndFor
		\State Send $c_{i,r}\gets\myenc_{pk}(Y_{i, r})$ to auxiliary server $r$
		\vspace{0.3cm}
		\setcounter{ALG@line}{0}
		\AuxServerP{Shares $\tuple{Y_{i,j}}_{i \in [n]}$}
		\For{$k\in[n_r]$}\Comment{Generate shares of fake reports}
		\State Sample $Y'_{k,j}$ uniformly from output space of \fo
		\EndFor
		\State Participate in \OSHA with $\tuple{Y_{i,j}}_{i \in [n]}$ and $\tuple{Y'_{k,j}}_{k\in[n_r]}$ and send the shuffled result to the server
		\vspace{0.3cm}
		\setcounter{ALG@line}{0}
		\AuxServerR{Encrypted shares $\tuple{c_{i,r}}_{i \in [n]}$}
		\For{$k\in[n_r]$}\Comment{Encrypted shares of fake reports}
		\State Sample $Y'_{k,r}$ uniformly from output space of \fo
		\State $c'_{k,r}\gets\myenc_{pk}(Y'_{k, r})$
		\EndFor
		\State Participate in \OSHA with $\tuple{c_{i,r}}_{i \in [n]}$ and $\tuple{c'_{k,r}}_{k\in[n_r]}$ and send the shuffled result to the server
		\vspace{0.3cm}
		\setcounter{ALG@line}{0}
		\Server{Shares from auxiliary servers}
		\State Decrypt and aggregate the shares to recover $Y$
		\State For any $v\in \Domain$, estimate ${f}'_v$ using $Y$ and Equation~\eqref{eq:murs_fo}
	\end{algorithmic}
\end{algorithm}

\subsection{Privacy Analysis}
\label{subsec:murs_privacy}

Now we analyze the privacy guarantee of \methodsmc.  Because of the usage \OSHA protocol, the server knows all the fake reports and each user's LDP report if it can corrupt more than $\lfloor r/2 \rfloor$ of the shufflers.  And in this case, each user's privacy is only protected by $\epsilon_l$-DP.
On the other hand, as long as the server cannot corrupt more than $\lfloor r/2 \rfloor$ shufflers, the server cannot gain useful information.

In what follows, we assume the server cannot corrupt more than $\lfloor r/2 \rfloor$ shufflers and examine the privacy guarantee of \methodsmc.
The focus is on how the privacy guarantees change after the addition of $n_r$ fake reports.
With these injected reports, what the server can observe is the reports from both users and the shufflers.
If the users collude, the server can subtract all other users' contribution and the privacy comes from the fake reports.  The following corollaries give the precise privacy guarantee:

\begin{cor}
	\label{cor:fr_privacy}
	If \slh is used and \slh is $\epsilon_l$-LDP, then \methodsmc is $\epsilon_c$-DP against the server; and if other users collude with the server, the protocol is $\epsilon_s$-DP, where
	\begin{align}
		\epsilon_s & = \sqrt{14\ln (2/\delta) \cdot \frac{d'}{n_r}}\nonumber                                                           \\
		\epsilon_c & = \sqrt{14\ln (2/\delta) / \left( \frac{n-1}{e^{\epsilon_l} +d' -  1} + \frac{n_r}{d'}\right)}\label{eq:eps_fn_c}
	\end{align}

\end{cor}

\begin{proof}
	The proof is similar to the setting of with \slh, but with $n_r$ more random reports.
	More specifically, when other users collude, privacy is provided by the $n_r$ random reports that are always random, and follow uniform distribution over $[d']$.  Plugging the argument into Equation~\eqref{eq:proof_ratio}, these can be viewed as a random variable that follows Binomial distribution with $\mathsf{Bin}\left(n_r, \frac{1}{d'}\right)$.
	The rest of the proof follows from that for Theorem~\ref{thm:olh_privacy}.

	Similarly, for the privacy guarantee against the server, there are $n-1$ random reports from users, and $n_r$ reports from the auxiliary server.  The effect of both can be viewed as one Binomial random variable: $\mathsf{Bin}\left(n - 1, 1/(e^{\epsilon_l}+d'-1)\right) + \mathsf{Bin}\left(n_r, 1/d'\right) = \mathsf{Bin}\left(n - 1 + n_r, \frac{(n-1)/(e^{\epsilon_l}+d'-1) + n_r/d'}{n - 1 + n_r}\right)$.

\end{proof}

One can also use \grr in \methodsmc, and we have a similar theorem:

\begin{cor}
	\label{cor:fr_privacy_grr}
	If \grr is used and \grr is $\epsilon_l$-LDP, then \methodsmc is $\epsilon_c$-DP against the server; and if other users collude with the server, the protocol is $\epsilon_s$-DP, where
	\begin{align}
		\epsilon_s &= \sqrt{14\ln (2/\delta) \cdot \frac{d}{n_r}}\nonumber\\\epsilon_c &= \sqrt{14\ln (2/\delta) / \left( \frac{n-1}{e^{\epsilon_l} +d -  1} + \frac{n_r}{d}\right)}\nonumber \end{align}
\end{cor}
The proof is similar to that for Corollary~\ref{cor:fr_privacy} and is thus omitted.

\subsection{Utility Analysis}
\label{subsec:murs_utility}
In Section~\ref{subsec:utility_basic}, we analyze the accuracy performance of different methods under the basic shuffling setting.  In this section, we further analyze the utility of these methods in \methodsmc.  The difference mainly comes from the fact that $n_r$ dummy reports are inserted, and the server runs a further step (i.e., Equation~\eqref{eq:murs_fo}) to post-process the results.  In what follows, we first show that Equation~\eqref{eq:murs_fo} gives an unbiased estimation; based on that, we then provide a general form of estimation accuracy.

We first show $f'_v$ is an unbiased estimation of $f_v$, where $f_v = \frac{1}{n}\sum_{i\in[n]} \ind{v_i=v}$.

\begin{lemma}
	\label{lem:fr_olh_unbias}
	The server's estimation $f'_v$ from Equation~\eqref{eq:murs_fo} is an unbiased estimation of $f_v$, i.e.,
	\begin{align*}
		\EV{\tilde{f}_v} = f_v
	\end{align*}
\end{lemma}

\begin{proof}
	\begin{align}
		\EV{f'_v} = & \EV{\frac{n + n_r}{n}\tilde{f}_v - \frac{n_r}{n}\frac{1}{d}} \nonumber                 \\
		=           & \frac{n + n_r}{n}\EV{\tilde{f}_v} - \frac{n_r}{n}\frac{1}{d}\label{eq:est_expectation}
	\end{align}
	Here $\tilde{f}_v$ is the estimated frequency of value $v$ given the $n+n_r$ reports; among them, $n$ of them are from the true users, and $n_r$ are from the randomly sampled values.  For the $n$ reports from users, $nf_v$ of them have original value $v$; and for the $n_r$ reports, in expectation, $n_r/d$ of them have original value $v$.  After perturbation, we have
	\begin{align*}
		\EV{\tilde{f}_v} = \frac{nf_v + n_r/d}{n + n_r}
	\end{align*}
	Putting it back to Equation~\eqref{eq:est_expectation}, we have $\EV{\tilde{f}_v} = f_v$.
\end{proof}

Given that, we prove the expected squared error of $f'_v$:
\begin{align}
	\Var{f'_v} = \Var{\frac{n+n_r}{n}\tilde{f}_v - \frac{n_r}{n}\frac{1}{d}}
	=\frac{(n+n_r)^2}{n^2}\Var{\tilde{f}_v}\nonumber
\end{align}
Now plugging in the results of $\Var{\tilde{f}_v}$ from Section~\ref{subsec:utility_basic} (note that we use replace $n$ with $n+n_r$ in the denominator as there are $n+n_r$ total reports), we obtain the specific variance of different methods after inserting $n_r$ dummy reports.

Corollary~\ref{cor:fr_privacy} gives both $\epsilon_s$ and $\epsilon_c$.  For $\epsilon_s$, $d'$ is fixed given $n_r$ and $\delta$; but we can vary $d'$ given $\epsilon_c$.  In particular, we can also derive the optimal value of $d'$ following the similar to the analysis of Section~\ref{subsec:utility_basic} (after Proposition~\ref{thm:utility_slh}):

Given $\epsilon_c = \sqrt{14\ln (2/\delta) / \left( \frac{n-1}{e^{\epsilon_l} +d' -  1} + \frac{n_r}{d'}\right)}$, we have
\begin{align*}
	e^{\epsilon_l} +d' -  1 = \frac{n-1}{14\ln (2/\delta) /\epsilon_c^2 - n_r/d'}
\end{align*}
We denote it as $m$, and (to simplify the notations) use $a$ to represent $14\ln (2/\delta) /\epsilon_c^2$ and $b$ to represent $n-1$.  By the variance derived above, we have $\var = \frac{m^2}{(m-d)^2(d-1)}\frac{n+n_r}{n^2}$.  Note that this formula is similar to the previous one in Section~\ref{subsec:utility_basic}; but here $m$ also depends on $d'$.  Thus we need to further simplify $\var$:
\begin{align*}
	\var = & \frac{(n+n_r)\left(\frac{b}{a - n_r/d'}\right)^2}{n^2\left(\frac{b}{a - n_r/d'} - d\right)^2(d'-1)} \\
	=      & \frac{(n+n_r)b^2}{n^2\left(b - (a - n_r/d)d'\right)^2(d'-1)}                                        \\
	=      & \frac{(n+n_r)b^2}{n^2a^2\left(d' - (b + n_r)/a\right)^2(d'-1)}
\end{align*}
To minimize $\var$, we want to maximize $\left(d' - (b + n_r)/a\right)^2(d'-1)$.  By making its partial derivative to $0$, we can obtain that when
\begin{align*}
	d'=\frac{(b + n_r)/a+2}{3}=\frac{\epsilon_c^2(n - 1 - n_r)}{42\ln (2/\delta)} + \frac{2}{3}
\end{align*}
the variance is minimized.  Comparing to Equation~\eqref{eq:opt_d}, introducing $n_r$ will reduce the optimal $d'$.  We use the integer component of $d'$ in the actual implementation.

\subsection{Discussion and Guideline}

\methodsmc strengthens the security aspect of the shuffler model from three perspectives:  First, it provides better privacy guarantee when users collude with the server, which is a common assumption made in DP.  Second, it makes the threat of the server colluding with the shufflers more difficult.  Third, it limits the ability of data poisoning of the shufflers.  We discuss criteria for initiating \methodsmc.

\mypara{Choosing Parameters}
Given the desired privacy level $\epsilon_1, \epsilon_2, \epsilon_3$ against the three adversaries $\Adv, \Adv_u, \Adv_a$, respectively.  Also given the domain size $d$, number of users $n$, and $\delta$, we want to configure \methodsmc so that it provides $\epsilon_c\leq \epsilon_1$, $\epsilon_s\leq \epsilon_2$, and $\epsilon_l\leq \epsilon_3$.

Local perturbation is necessary to satisfy $\epsilon_3$-DP against $\Adv_a$.  To achieve $\epsilon_2$ when other users collude, noise from auxiliary servers are also necessary.  Given that, to satisfy $\epsilon_c\le \epsilon_1$, if we have to add more noise, we have two choices.  That is, the natural way is to add noisy reports from the auxiliary server, but we can also lower $\epsilon_l$ at the same time.  As we have the privacy and utility expressions, we can numerically search the optimal configuration of $n_r$ and $\epsilon_l$.  Finally, given $\epsilon_l$, we can choose to use either \grr or \slh by comparing Theorem~\ref{thm:olh_privacy} and Theorem~\ref{thm:utility_grr}.

\section{Evaluation}
\label{sec:experiments}

The purpose of the evaluation is two-fold.  First, we want to measure the utility of \slh, i.e., how much it improves over exsiting work.  Second, we want to measure the communication and computation overhead of \methodsmc, to see whether the technique is applicable in practice.

As a highlight, our \methodsmc can make estimations that has absolute errors of $<0.01\%$ in reasonable settings, improving orders of magnitude over existing work.  The overhead is small and practical.

\subsection{Experimental Setup}\label{ssec:experimental_setup}

\mypara{Datasets}
We run experiments on three real datasets.

\begin{itemize}
	\item
	      IPUMS~\cite{data:ipums}: The US Census data for the year 1940.  We sample $1\%$ of users, and use the city attribute (N/A are discarded).  This results in $n=602325$ users and $d=915$ cities.

	\item
	      Kosarak~\cite{data:FIMI}: A dataset of $1$ million click streams on a Hungarian website that contains around one million users with $42178$ possible values.  For each stream, one item is randomly chosen.

	\item
	      AOL~\cite{data:aol}: The AOL dataset contains user queries on AOL website during the first three months in 2006. We assume each user reports one query (w.l.o.g., the first query), and limit them to be 6-byte long. This results a dataset of around 0.5 million queries including 0.12 million unique ones.  It is used in the succinct histogram case study in Section~\ref{sec:heavy_hitter}.
\end{itemize}

\mypara{Competitors}
We compare the following methods:
\begin{itemize}

	\item
	      \olh: The local hashing method with the optimal $d'$ in the LDP setting~\cite{wang2017locally}.

	\item
	      \had: The Hadamard transform method used in~\cite{aistats:AcharyaSZ18}.  It can be seen as \olh with $d'=2$ (utility is worse than \olh); but compared to \olh, its server-side evaluation is faster.

	\item
	      \MN: The shuffler-based method for histogram estimation~\cite{balle2019privacy}.

	\item
	      \aue: Method from~\cite{Balcer2019}.  It first transforms each user's value using one-hot encoding.  Then the values ($0$ or $1$) in each location is incremented w/p $p = 1 - \frac{200}{\epsilon_c^2n}\ln(4/\delta)$.  Note that it is not an LDP protocol, and its communication cost is $O(d)$.

	\item
	      \rap: The hashing-based idea described in Section~\ref{sec:unary_encoding}.  Its local side method is equivalent to RAPPOR~\cite{rappor}.
	      Similar to \aue, it has large communication cost.

	\item
	      \rapr: Method from~\cite{arxiv:erlingsson2020encode}.  Similar to \aue and \rap, it transforms each user's value using one-hot encoding.  The method works in the removal setting of DP.  When converting to the replacement definition, it has the same utility as \rap.

	\item
	      \slh: The hashing-based idea introduced in Section~\ref{subsec:slh}.

	\item
	      \methodsmc:
	      We focus on the perspective of the computation and communication complexity in Section~\ref{subsec:exp_complexity}.

	\item
	      \methodseq: As a baseline, we also evaluate the complexity of the sequential shuffling method presented in~\ref{subsec:ss}; we call it \methodseq.
\end{itemize}

\mypara{Implementation}
The prototype was implemented using Python 3.6 with fastecdsa 1.7.4, pycrypto 2.6.1, python-xxhash 1.3.0 and numpy 1.15.3 libraries. For \methodseq, we generate a random AES key to encrypted the message using AES-128-CBC, and use the ElGamal encryption with elliptic curve secp256r1 to encrypt the AES key. For the AHE in \methodsmc, we use DGK \cite{damgaard2007efficient} with 3072-bits ciphertext.  All of the encryption used satisfy 128-bit security.

\mypara{Metrics}
We use mean squared error (\mse) of the estimates as metrics.  For each value $v$, we compute its estimated frequency $\tilde{f}_v$ and the ground truth $f_v$, and calculate their squared difference.  Specifically, $\mse = \frac{1}{|\Domain|} \sum_{v \in \Domain} (f_v - \tilde{f}_v)^2 $.

\mypara{Methodology}
For each dataset and each method, we repeat the experiment $100$ times, with result mean and standard deviation reported.  The standard deviation is typically very small, and barely noticeable in the figures. By default, we set $\delta=10^{-9}$.

\subsection{Frequency Estimation Comparison}
We first show the utility performance of \slh.  We mainly compare it against other methods in the shuffler model, including \MN, \aue, \rap, and \rapr.  For comparison, we also evaluate several kinds of baselines, including LDP methods \olh and \had, centralized DP method Laplace mechanism (Lap) that represents the lower bound, and a method Base that always outputs a uniform distribution.

Figure~\ref{fig:mse_epsc} shows the utility comparison of the methods.  We vary the overall privacy guarantee $\epsilon_c$ against the server from $0.1$ to $1$, and plot \mse.
First of all, there is no privacy amplification for \MN when $\epsilon_c$ is below a threshold.  In particular, when $\epsilon_c < \sqrt{\frac{14\ln (2/\delta)d}{n-1}}$, $\epsilon_l=\epsilon_c$.  We only show results on the IPUMS dataset because for the Kosarak dataset, $d$ is too large so that \MN cannot benefit from amplification.  When there is no amplification, the utility of \MN is poor, even worse than the random guess baseline method.
Compared to \MN, our improved \slh method can always enjoy the privacy amplification advantage, and gets better utility result, especially when $\epsilon_c$ is small.
The three unary-encoding-based methods \aue, \rap, and \rapr are all performing similar to \slh.  But the communication cost of them are higher.  The best-performing method is \rapr; but it works in the removal-LDP setting.  Because of this, its performance with $\epsilon_c$ is equivalent to \rap with $2\epsilon_c$.

Moving to the LDP methods, \olh and \had perform very similar (because in these settings, \olh mostly chooses $d'=2$ or $3$, which makes it almost the same as \had), and are around $3$ orders of magnitude worse than the shuffler-based methods.
For the central DP methods, we observe Lap outperforms the shuffler-based methods by around $2$ orders of magnitude.

\begin{figure}[t]
	\centering

	{
		\includegraphics[width=0.47\textwidth]{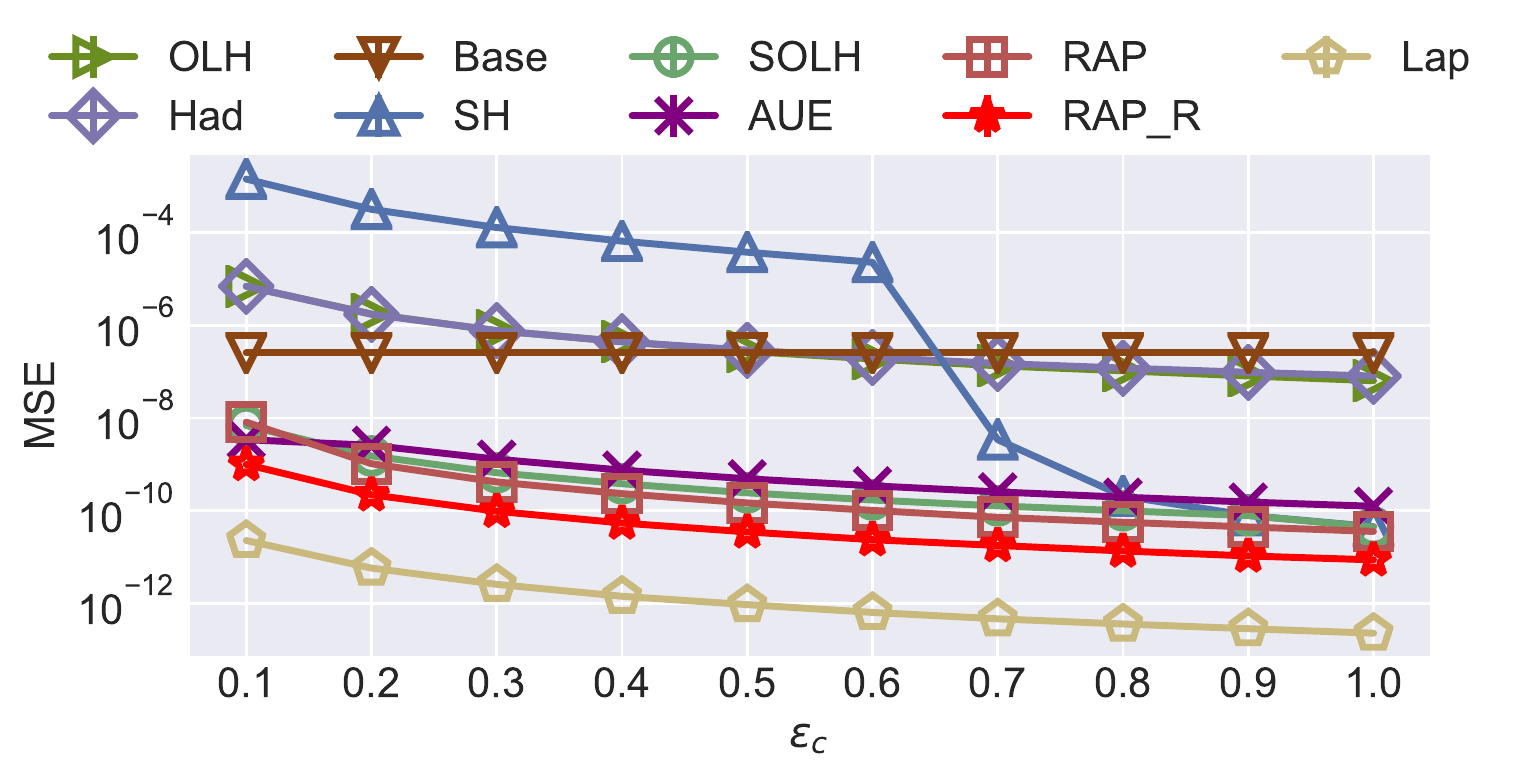}
	}
	\caption{
		Results of \mse varying $\epsilon_c$ on the IPUMS dataset. Base always outputs $1/d$ for each estimation.  Lap stands for Laplace mechanism for DP.
	}
	\label{fig:mse_epsc}
\end{figure}

\begin{table}[h]
	\centering
	\resizebox{0.99\columnwidth}{!}{
		\begin{tabular}{@{}c|c|c|c|c|c@{}}
			\toprule
			Metric                   & \backslashbox{Method}{$\epsilon_c$} & 0.2     & 0.4     & 0.6      & 0.8      \\
			\midrule
			$d'$                     & \slh                                & 45      & 177     & 397      & 705      \\\hline
			\multirow{5}{*}{Utility} & \slh                                & 5.27e-8 & 1.30e-8 & 5.76e-9  & 3.24e-9  \\
			                         & \rapr$(d'=10)$                      & 1.31e-7 & 1.17e-7 & 1.14e-7  & 1.13e-7  \\
			                         & \rapr$(d'=100)$                     & 1.73e-7 & 1.55e-8 & 1.22e-8  & 1.22e-8  \\
			                         & \rapr$(d'=1000)$                    & 1.02e-4 & 2.60e-5 & 4.02e-8  & 3.66e-9  \\
			                         & \rapr                               & 7.82e-9 & 1.92e-9 & 8.53e-10 & 4.78e-10 \\
			\bottomrule
		\end{tabular}
	}
	\caption{
		Comparison of \slh and \rapr in Kosarak.}
	\label{tbl:d_prime}
\end{table}

In Table~\ref{tbl:d_prime}, we list the value of $d'$ of \slh and the utility of \slh and \rapr for some $\epsilon_c$ values.  We also fix $d'$ in \slh and show how sub-optimal choice of $d'$ makes \slh less accurate.  The original domain $d$ is more than $40$ thousand, thus introducing a large communication cost compared to \slh ($5$KB vs $8$B).  The computation cost for the users is low for both methods; but for the server, estimating frequency with \slh requires evaluating hash functions.  We note that as this takes place on server, some computational cost is tolerable, especially the hashing evaluation nowadays is efficient.  For example, our machine can evaluate the hash function 1 million times within 0.1 second on a single thread.

\subsection{Succinct Histograms}
\label{sec:heavy_hitter}
In this section, we apply shuffle model to the problem of succinct histogram (e.g.,~\cite{stoc:BassilyS15,nips:BassilyNST17}) as a case study.
The succinct histogram problem still outputs the frequency estimation; but different from the ordinary frequency or histogram estimation problem, which we focused on in the last section, it handles the additional challenge of a much larger domain (e.g., domain size greater than $2^{32}$).
To deal with this challenge, \cite{nips:BassilyNST17} proposes TreeHist.
It assumes the domain to be composed of fixed-length binary strings and constructs a binary prefix tree.
The root of the tree denotes the empty string.  Each node has two children that append the parent string by $0$ and $1$.  For example, the children of root are two prefixes $0*$ and $1*$, and the grand children of root are $00*$, $01*, 10*$, and $11*$.  The leaf nodes represent all possible strings in the domain.

To find the frequent strings, the algorithm traverses the tree in a breadth-first-search style:  It starts from the root and checks whether the prefixes at its children are frequent enough.  If a prefix is frequent, its children will be checked in the next round.  For each round of checking, an LDP mechanism (such as those listed in Section~\ref{subsec:ldp}) is used.
Note that the mechanism can group all nodes in the same layer into a new domain (smaller than the original domain because many nodes will be infrequent and ignored).  Each user will check which prefix matches the private value, and report it (or a dummy value if there is no match).
In this section, to demonstrate the utility gain of the shuffler model, we use the methods \MN, \slh, \aue, and \rap as the frequency estimator (i.e., the framework of TreeHist stays the same; but the frequency estimator is changed).

In what follows, we empirically compare them to demonstrate the applicability and benefit of the shuffler model.  Following the setting of~\cite{nips:BassilyNST17}, we consider the AOL dataset assuming each user's value is $48$ bits.  We run TreeHist in $6$ rounds, each for $8$ bits ($1$ character).  We set the goal to identify the the top $32$ strings, and in each intermediate round, we identify the top $32$ prefixes.  In the LDP setting, TreeHist divides the users into $6$ groups, as that gives better results.  In the shuffler case, a better approach is to avoid grouping users, but rather dividing $\epsilon_c$ and $\delta_c$ by $6$ for each round.

Figure~\ref{fig:heavy_hitter} shows the results.  We can observe that the except \MN, the other shuffler-based methods outperforms the LDP TreeHist (\olh and \had) .  In addition to the capability of reducing communication cost, another advantage of \slh we observe here is that \slh enables non-interactive execution of TreeHist (note that this is also one reason why the original TreeHist algorithm uses the local hashing idea).
In particular, the users can encode all their prefixes and report together.  The server, after obtaining some frequent prefix, can directly test the potential strings in the next round.
On the other hand, using the unary-encoding-based methods, users cannot directly upload all their prefixes, because the size of a report can be up to $2^{48}$ bits.  Instead, the server has to indicate which prefixes are frequent to the users and then request the users to upload.

\begin{figure}[t]
	\centering
	{
		\includegraphics[width=0.47\textwidth]{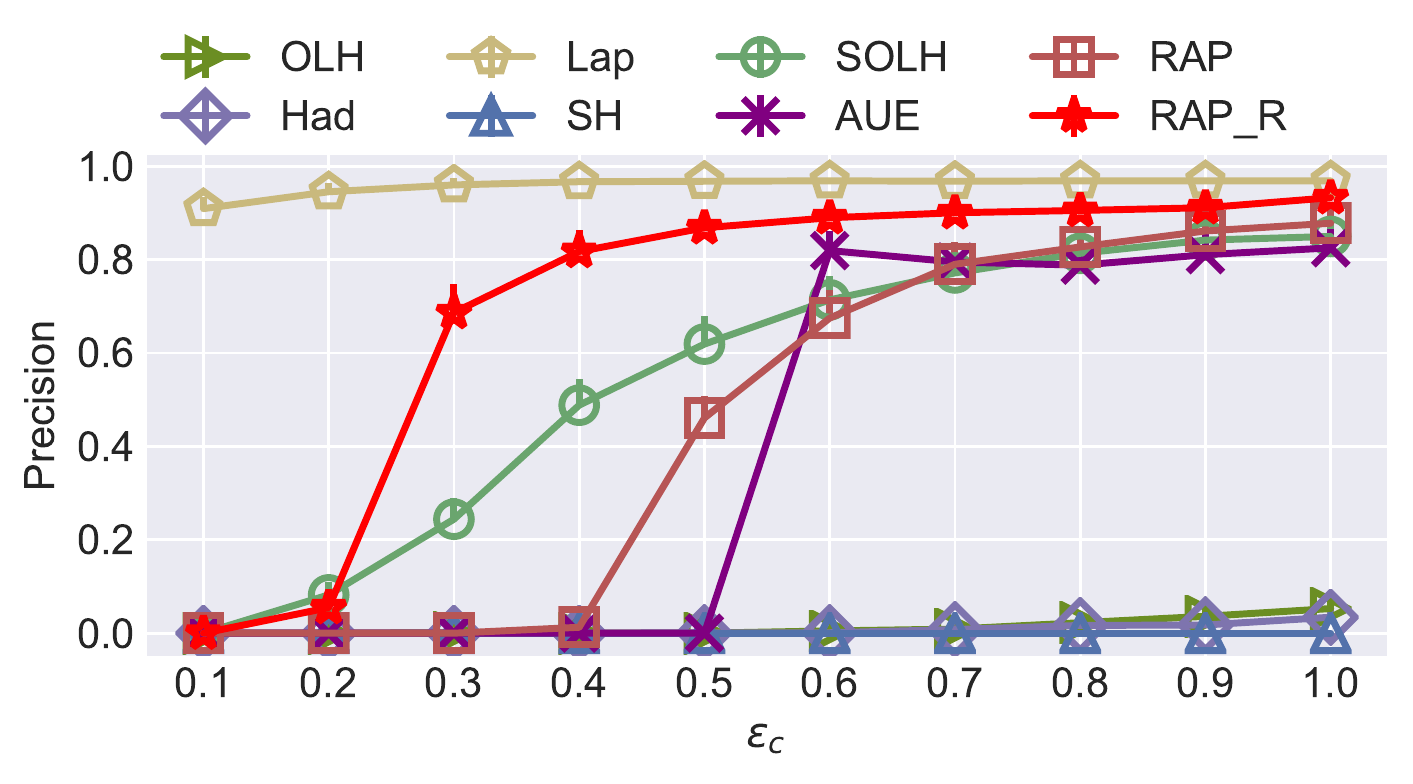}
	}
	\caption{Comparison on the succinct histogram problem.  The target is to identify the top $32$ most frequent values.}
	\label{fig:heavy_hitter}
\end{figure}

\subsection{Performance Evaluation}
\label{subsec:exp_complexity}
We evaluate the computational and communication costs of \methodseq and \methodsmc, focusing on the overhead introduced by the encryption and shuffling.
We run the experiments on servers running Linux kernel version 5.0 with Intel Xeon Silver 4108 CPU @ 1.80GHz and 128GB memory.  We assume there are $r=3$ and $r=7$ shufflers.  The results are listed in Table~\ref{tbl:complexity}.
As both methods scales with $n+n_r$, we fix $n$ to be $1$ million and ignore $n_r$.

Note that we the results are only for \slh with report size fixed at $64$ bits.  If we use \rap in this case, the communication cost will increase proportional to the size of the domain $d$ (by $d/64$).

\begin{table}[h]
	\centering
	\resizebox{0.9\columnwidth}{!}{
		\begin{tabular}{@{}c|c|c|c|c@{}}
			\toprule
			\multirow{2}{*}{\backslashbox{Metric}{Method} } & \multicolumn{2}{c|}{\methodseq} & \multicolumn{2}{c}{\methodsmc}                  \\
			                                                & $r=3$                           & $r=7$                          & $r=3$ & $r=7$  \\ \midrule
			User comp. (ms)                                 & 0.24                            & 0.49                           & 1.6   & 1.6    \\
			User comm. (Byte)                               & 416                             & 800                            & 400   & 432    \\ \hline
			Aux. comp. (s)                                  & 49                              & 50                             & 0.2   & 0.7    \\
			Aux. comm. (MB)                                 & 224                             & 416                            & 429.8 & 3293.3 \\ \hline
			Server comp. (s)                                & 49                              & 49                             & 65    & 65     \\
			Server comm. (MB)                               & 128                             & 128                            & 392   & 408    \\ \bottomrule
		\end{tabular}
	}
	\caption{Computation and communication overhead of \methodseq and \methodsmc for each user, each shuffler, and the server.  We assume $n=10^6$ and $r=3$ or $7$.}
	\label{tbl:complexity}
\end{table}

\mypara{User Overhead}
Overall, the user-side computation and communication overhead are small for both methods.  The computation only involves sampling, secret-sharing, and $r$ times of encryption operations.  All of them are fast.  Note that in \methodseq, as onion encryption is used, its overhead is larger and grows linearly with respect to $r$.  The communication cost for each user is also very limited.

\mypara{Shuffler Overhead}
For each shuffler in \methodseq, the computation cost lies in $n$ decryptions (for one layer), sampling $n_u$ random reports (with necessary encryption), and then shuffling.  Note that the decryptions is done in parallel.  In this implementation, we use $32$ threads for demonstration.  With more resources, the processing time can be shortened.

In \methodseq, an ElGamal ciphertext is a tuple $\tuple{P, C}$,  $P$ is a point in the secp256r1 curve and thus can be represented by $256\times 2$ bits; and $C$ is a number in $\{0, 1\}^{256}$. That is, we need $96$ bytes to store the encrypted AES key for each layer.
For \slh, we let each user randomly select an $4$-byte seed as the random hash function.  After padding, each message is $32 + 96(r+1)$ bytes, where $r$ is the number of layers used for shufflers.  One additional layer is used for the server.
Given $n=1$ million users and $r$ shufflers, there will be on average $\frac{1}{r}\times n\times \sum_{k=1}^r(32+96(k+1)) = 672$ MB data sent to the three shufflers.

\methodsmc is made up of ${r \choose {\lfloor r/2 \rfloor + 1}}$ rounds of sorting.  Since a well-implemented sorting on $1$ million elements takes only several milliseconds, the computation cost of shuffling is minor for the shufflers.  In addition, our protocol require each shuffler do ${r \choose {\lfloor r/2 \rfloor + 1}}\cdot n / r$ homomorphic additions during shuffling.  As Table~\ref{tbl:complexity} indicates, all of these cryptographic operations are efficient.  The cost is no more than one second with $n = 1$ million reports.

According to the analysis of oblivious shuffle from~\cite{laur2011round}, each shuffler's communication cost is $O(2^r\sqrt{r}n)$. In addition, our protocol sends $n$ encrypted shares each round, which introduces another communication cost of $O(2^rn/\sqrt{r})$ by similar analysis (multiplied with a larger constant factor because of the 3072-bit DGK ciphertexts).
In experiments with $1$ million users and 3 shufflers, each shuffler needs to send 430 MB. In a more expensive case with 7 shufflers, it becomes 3.3 GB.  While the communication cost is higher than that of \methodseq, we note that the cost is tolerable in our setting, as the data collection does not happen frequently.

\mypara{Server Overhead}
For \methodseq, the server computation overhead is similar to that of the shufflers, as they all decrypt one layer.  The server's communication cost (measured by amount of data received) is lower though, as there is only one layer of encryption on the data.

In \methodsmc, the server needs to collect data from all $r$ shufflers.  As only one share is encrypted by DGK, the communication overhead is mostly composed of that part and grows slowly with $r$.  The computation overhead is also dominated by decrypting the DGK ciphertexts.

\section{Related Work}
\label{sec:related}

\mypara{Privacy Amplification by Shuffling}
The shuffling idea was originally proposed in Prochlo~\cite{bittau2017prochlo}.  Later the formal proof was given in~\cite{erlingsson2019amplification,cheu2018distributed,balle2019privacy}.
Parallel to our work, \cite{Balcer2019,ghazi2019power} propose mechanisms to improve utility in this model.  They both rely on the privacy blanket idea~\cite{balle2019privacy}.
More recently, \cite{arxiv:erlingsson2020encode} considered an intriguing removal-based LDP definition and work in the shuffler model.
Besides estimating histograms, the problem of estimating the sum of numerical values are also extensively investigated~\cite{ghazi2019private,ghazi2019scalable,balle2020private}.

\mypara{Crypto-aided Differential Privacy}
Different from using shufflers, researchers also proposed methods that utilize cryptography to provide differential privacy guarantees, including~\cite{froelicher2017unlynx,elahi2014privex,melis2015efficient}.
One notable highlight is \cite{chowdhury2019outis}, which proposes Crypt$\epsilon$.  In this approach, users encrypt their values using homomorphic encryption, and send them to the auxiliary party via a secure channel.  The auxiliary server tallies the ciphertext and adds random noise in a way that satisfies centralized DP, and sends the result to the server.  The server decrypts the aggregated ciphertext.  More recently, researchers in~\cite{roth2019honeycrisp} introduce several security features including verification and malice detection.
This line of work does not require LDP protection, thus differs from our approach.  Moreover, to handle the histogram estimation when $|\Domain|$ is larger, the communication overhead is larger than that of ours.

\mypara{Relaxed Definitions}
Rather than introducing the shuffler, another direction to boost the utility of LDP is to relax its \textit{semantic meaning}.
In particular, Wang et al. propose to relax the definition by taking into account the distance between the true value and the perturbed value~\cite{wang2017local_ordinal}.  More formally, given the true value, with high probability, it will be perturbed to a nearby value (with some pre-defined distance function); and with low probability, it will be changed to a value that is far apart.  A similar definition is proposed in \cite{gursoy2019secure,gu2019supporting}.  Both usages are similar to the geo-indistinguishability notion in the centralized setting~\cite{andres2012geo}.
In~\cite{murakami2019utility}, the authors consider the setting where some answers are sensitive while some not (there is also a DP counterpart called One-sided DP~\cite{doudalis2017one}).  The work~\cite{gu2019providing} is a more general definition that allows different values to have different privcay level.  Our work applied to the standard LDP definition, and we conjecture that these definitions can also benefit from introducing a shuffler without much effort.

There also exist relaxed models that seem incompatible with the shuffler model, i.e., \cite{bhowmick2018protection} considers the inferring probability as the adversary's power; and \cite{sigmod:WangDZHHLJ19} utilizes the linkage between each user's sensitive and public attributes.

\mypara{Distributed DP}
In the distributed setting of DP, each data owner (or proxy) has access to a (disjoint) subset of users.  For example, each patient's information is possessed by a hospital.  The DP noise is added at the level of the intermediate data owners (e.g.,~\cite{mcmahan2017federated}).
A special case (two-party computation) is also considered~\cite{he2017composing,rao2019hybrid}.
\cite{mcgregor2010limits} studies the limitation of two-party DP.
In~\cite{DKMMN06}, a distributed noise generation protocol was proposed to prevent some party from adding malicious noise.  The protocol is then improved by~\cite{champion2019securely}.
\cite{mironov2009computational} lays the theoretical foundation of the relationship among several kinds of computational DP definitions.

We consider a different setting where the data are held by each individual users, and there are two parties that collaboratively compute some aggregation information about the users.

\mypara{DP by Trusted Hardware}
In this approach, a trusted hardware (e.g., SGX) is utilized to collect data, tally the data, and add the noise within the protected hardware.  The result is then sent to the analyst.  Google propose Prochlo~\cite{bittau2017prochlo} that uses SGX.  Note that the trusted hardware can be run by the server.
Thus \cite{chan2019foundations} and~\cite{allen2019algorithmic} designed oblivious DP algorithms to overcome the threat of side information (memory access pattern may be related to the underlying data).
These proposals assume the trusted hardware is safe to use.  However, using trusted hardware has potential risks (e.g.,~\cite{biondo2018guard}).  This paper considers the setting without trusted hardware.

\section{Conclusions}
\label{sec:conc}
In this paper, we study the shuffler model of differential privacy from two perspectives.  First, we examine from the algorithmic aspect, and make improvement to existing techniques.  Second, we work from the security aspect of the model, and emphasize two types of attack, collusion attack and data-poisoning attack; we then propose \methodsmc that is safer under these attacks.  Finally, we perform experiments to compare different methods and demonstrate the advantage of our proposed method.

In summary, we improve both the utility and the security aspects of the shuffler model.  For the problem of histogram estimation, our proposed protocol is both more accurate and more secure than existing work, with a reasonable communication/computation overhead.  We also demonstrate the applicability of our results in the succinct histogram problem.

	{
		\bibliographystyle{IEEEtranS}
		\bibliography{bibs/abbrev0,bibs/Ninghui,bibs/privacy}
	}

\appendix
\section{Appendices}

\begin{proof}	
	Denote $\AA$ as the algorithm of \slh in the shuffler model.  Let $\AA(D) =  [\slh(v_{\pi(1)}), \ldots, \slh(v_{\pi(n)})]$ be the output on a dataset $D$, where $\pi$ is a random permutation from $[n]$ to $[n]$.
	W.l.o.g., we assume $D$ and $D'$ differ in the $n$-th value, i.e., $v_n\neq v'_n$.  We denote $\result$ as the output from $\AA(D)$.  It is of the form
	$[\tuple{H_j,y_j}]_{j\in[n]}$.
	To prove $\AA$ is $(\epsilon_c, \delta)$-DP, it suffices to show
	\begin{align*}
		 & \mathsf{Pr}_{{\result} \sim \AA(D)}\left[
			\frac{\Pr{\AA(D)=\result}}{\Pr{\AA(D')=\result}} \ge e^{\epsilon_c}
			\right] \leq \delta
	\end{align*}
	where the randomness is on coin tosses of all users' LDP mechanism and the shuffler's random shuffle.
	We first consider the algorithm $\AA(D)$ that, besides $\result$, also outputs two other values $T$ and $R_T$, where $T$ indicates the indices of the first $n-1$ users who report truthfully (i.e., with probability $1-\gamma = \frac{e^{\epsilon_l}-1}{e^{\epsilon_l}+d'-1}$), and $R_T$ denotes their chosen hash functions and hashed results ($R_T=[\tuple{\hat{H}_i,\hat{y}_i}]_{i\in T}$).  We prove that this algorithm is $(\epsilon_c, \delta)$-DP.  Given that, by the post-processing property of DP, if $\AA(D)$ only outputs $\result$ (this can be seen as a post-processing step that drops $T, R_T$), it is also $(\epsilon_c, \delta)$-DP.  We assume user $n$ also report truthfully.  Notice that if user $n$ report randomly the two probabilities are the same and can be canceled out.

	We first examine $\Pr{\AA(D)=(\result,T,R_T)}$:
	\begin{align}
		  & \Pr{\AA(D)=(\result,T,R_T)} \nonumber
		\\
		= & \sum_\pi \Pr{\pi} \Pr{\AA(D)=(\result,T,R_T)\mid \pi} \nonumber
		\\[-0.3cm]
		= & \sum_{\pi} \Pr{\pi} \left(\underbrace{\prod_{i\in T} \Pr{H_{\pi(i)}} \ind{H_{\pi(i)}=\hat{H}_i\wedge y_{\pi(i)} = \hat{y}_i}}_{\text{reports from users in $T$}} \cdot \right. \label{eq:proof_likelihood}
		\\[-0.6cm]
		  & \left. \underbrace{\prod_{i \in [n-1] \setminus T} \Pr{H_{\pi(i)}} \frac{1}{d'}}_{\text{reports from users in $[n-1] \setminus T$}} \cdot \underbrace{\Pr{H_{\pi(n)}} \ind{H_{\pi(n)}(v_n) = y_{\pi(n)}}}_{\text{report from user $n$}}\right) \nonumber
	\end{align}
	$\Pr{\pi}$ denotes the probability a specific random permutation is chosen ($\Pr{\pi}=1/n!$), $\Pr{H_{\pi(i)}}$ is the probability user $i$ chooses hash function $H_{\pi(i)}$ (assuming there are $h$ possible hash functions, $\Pr{H_{\pi(i)}}=1/h$), and the summation is over all permutation $\pi$.  However, only some $\pi$ ensures $H_{\pi(i)}=\hat{H}_i$ and $y_{\pi(i)} = \hat{y}_i$ for $i\in T$.
	Such a $\pi$ always maps $i\in T$ to fixed locations ($H_{\pi(i)}=\hat{H}_i$ and $y_{\pi(i)} = \hat{y}_i$).  Denote $P=\{\pi\mid \forall i\in T, H_{\pi(i)}=\hat{H}_i\wedge y_{\pi(i)} = \hat{y}_i\}$, we have
	\begin{align}
		\frac{\Pr{\AA(D)=(\result,T,R_T)} }{\Pr{\AA(D')=(\result,T,R_T)} } = & \frac{
			c\sum_{\pi\in P} \ind{ H_{\pi(n)}(v_n) = y_{\pi(n)}}}{c\sum_{\pi\in P} \ind{ H_{\pi(n)}(v'_n) = y_{\pi(n)}}}\nonumber
	\end{align}
	where $c=\Pr{\pi}(\prod_{i \in [n]}\Pr{H_{\pi(i)}})(\prod_{i \in [n-1] \setminus T} \frac{1}{d'}) $ is a constant that does not depend on $v_n$ or $v'_n$, and the denominator is from Equation~\eqref{eq:proof_likelihood} with a similar analysis on $D'$.

	In what follows, we prove $\sum_{\pi\in P} \ind{ H_{\pi(n)}(v_n) = y_{\pi(n)}} = c' \sum_{i\in [n]\setminus T} \ind{\hat{H}_{i}(v_n) = \hat{y}_{i}}$ for some constant $c'$, where we use $\hat{H}_{i}, \hat{y}_{i}$ to denote the report of user $i$.
	Define $R_{-T}$ as reports from $[n]\setminus T$.  First we assume $R_T$ and $R_{-T}$ are non-overlapping (there can be overlaps within $R_T$ or $R_{-T}$).
	In this case, any $\pi\in P$ will ensure reports from $T$ ($[n]\setminus T$, respectively) are mapped within $R_T$ ($R_{-T}$, respectively).
	For any index $i\in [n]\setminus T$ that $n$ is mapped to, denote $c_T$ as the number of possible mappings within $R_T$, and $c_{-T}$ as the number of mappings (random permutations) within $R_{-T}$, there are $c'=c_T\cdot c_{-T} / (n-|T|)$ valid permutations from $P$.

	For the case when $R_T$ and $R_{-T}$ are overlapping
	(this is actually very unlikely as each report involves a randomly chosen hash function from a potentially large hash family, e.g., we use $32$ bits to denote the seed of the hash function in the experiment),
	permutations that map some indices from $T$ to the overlapped reports are also valid.  Denote $c_{T'}$ as the number of such permutations.
	For $i\in [n]\setminus T$, if $\hat{H}_i, \hat{y}_i$ does not appear in $R_T$, we have $c' = c_{T'} \cdot c_{-T}/(n-|T|)$.  If $\exists j\in T$, s.t., $\hat{H}_i, \hat{y}_i$ = $\hat{H}_j, \hat{y}_j$, $n$ can also be mapped to $j$.  Summing up all such $j$'s, we also have $c' = c_{T'} \cdot c_{-T}/(n-|T|)$ valid permutations.
	Combining the two cases, we have:
	\begin{align}
		\frac{\Pr{\AA(D)=(\result,T,R_T)} }{\Pr{\AA(D')=(\result,T,R_T)} }
		= \frac{ \sum_{i\in [n]\setminus T} \ind{\hat{H}_{i}(v_n) = \hat{y}_{i}}}{\sum_{i\in [n]\setminus T} \ind{\hat{H}_{i}(v'_n) = \hat{y}_{i}} }\label{eq:proof_ratio}
	\end{align}

	So far, we have proved that, fixing $\result, T$ and $R_T$, the ratio only depends on the numbers of reports that are random and matches $v_n$ and $v'_n$, respectively.  The high level idea is to show that knowing $T$ and $R_T$ fixes the permutation on values from $T$; and any valid permutation only shuffles values from $[n]\setminus T$ (informally, this can be thought of as the server removes reports from $T$).  Now define
	\begin{align*}
		 & N_{\result, T, R_T} = \sum_{i\in [n]\setminus T} \left(\ind{\hat{H}_{i}(v_n) = \hat{y}_{i}} \right)   \\
		\mbox{ and }
		 & N'_{\result, T, R_T} = \sum_{i\in [n]\setminus T} \left(\ind{\hat{H}_{i}(v'_n) = \hat{y}_{i}}\right),
	\end{align*}
	we want to prove
	\begin{align*}
		     & \mathsf{Pr}_{(\result,T,R_T) \sim \AA(D)}\left[
			\frac{\Pr{\AA(D)=(\result,T,R_T)}}{\Pr{\AA(D')=(\result,T,R_T)}} \ge e^{\epsilon_c}
			\right]                                                                                                                                                                  \\
		     & \mbox{(omit the $(\result,T,R_T) \sim \AA(D)$ part to simplify notations)}                                                                                        \\
		=    & \mathsf{Pr} \left[ \frac{N_{\result, T, R_T}}{N'_{\result, T, R_T}}\geq e^{\epsilon_c}\right]                                                                     \\
		\leq & 1 - \mathsf{Pr}\left[N_{\result, T, R_T} \leq \theta e^{\epsilon_c/2} \wedge N'_{\result, T, R_T}\geq \theta e^{-\epsilon_c/2} \right]                            \\
		\leq & \mathsf{Pr}\left[N_{\result, T, R_T} \geq \theta e^{\epsilon_c/2}\right] + \mathsf{Pr}\left[ N'_{\result, T, R_T}\leq \theta e^{-\epsilon_c/2} \right]\leq \delta
	\end{align*}where $\theta$ is some constant.  For $(\result,T,R_T)$ generated from a random run of $\AA(D)$, we can show $N_{\result, T, R_T}$ and $N'_{\result, T, R_T}$ follow Binomial distributions.  In particular, as we assumed user $n$ always report truth, there must be $H_n(v_n)=y_n$; the remaining $n-1$ users will first decide whether to report truthfully (i.e., with probability $(e^{\epsilon_l} -1) / (e^{\epsilon_l} + d' - 1)$), and if user $i$'s report $\tuple{H_i, y_i}$ is random, we have $\Pr{H_i(v_n) = y_i} = 1/d'$.  Each user's reporting process are thus modeled as two Bernoulli processes.  As a result, $N_{\result, T, R_T}$ follows the Binomial distribution $\mathsf{Bin}(n - 1, 1 / (e^{\epsilon_l} + d' - 1))$ plus a constant $1$.
	Similarly, $N'_{\result, T, R_T}\sim \mathsf{Bin}(n - 1, 1 / (e^{\epsilon_l} + d' - 1)) + \ind{H_n(v'_n)=y_n} \geq \mathsf{Bin}(n - 1, 1 / (e^{\epsilon_l} + d' - 1))$.
	The rest of the proof follows that in the later part of the proof of Theorem 3.1 from~\cite{balle2019privacy}.  The high-level idea is to use Chernoff bound to prove the two probabilities $\Pr{\mathsf{Bin}(n - 1, 1 / (e^{\epsilon_l} + d' - 1)) + 1\geq \theta e^{\epsilon_c/2}}$ as well as $\Pr{\mathsf{Bin}(n - 1, 1 / (e^{\epsilon_l} + d' - 1))\leq \theta e^{-\epsilon_c/2}}$ are equal to or smaller than $\delta/2$ when
	$\theta=\frac{n-1}{e^{\epsilon_l} + d' - 1}\geq \frac{14k \ln (2/\delta)}{(n-1)\epsilon_c^2}$.
\end{proof}

\end{document}